\documentclass[cleveref,numberwithinsect,a4paper,UKenglish,pdfa]{lipics-v2021}

\newif\ifisarxiv

\isarxivtrue   

\ifisarxiv
  \pdfoutput=1
  \hideLIPIcs
  \newcommand{\appref}[1]{Appendix~\ref{#1}}
\else
  \relatedversiondetails[linktext={arXiv:2408.16102}, cite=extendedversion]{Full version with appendices}{https://arxiv.org/abs/2408.16102}
  \newcommand{\appref}[1]{\cite[Appendix {#1}]{extendedversion}}
\fi

\title{Beyond Monads and Biproducts: A Uniform Interpretation of Parallelism in Intuitionistic Logic}
\titlerunning{Beyond Monads and Biproducts: A Uniform Interpretation of Parallelism in IL}

\author{Alejandro Díaz-Caro}{Université de Lorraine, CNRS, Inria, LORIA, France\and Universidad Nacional de Quilmes, Argentina \and \url{https://members.loria.fr/ADiazCaro/} }{alejandro@diaz-caro.info}{https://orcid.org/0000-0002-5175-6882}{}

\author{Octavio Malherbe}{Universidad de la República, Facultad de Ingeniería, IMERL, Uruguay}{malherbe@fing.edu.uy}{}{}

\authorrunning{A. Díaz-Caro and O. Malherbe} 

\Copyright{Alejandro Díaz-Caro and Octavio Malherbe} 

\ccsdesc[500]{Theory of computation~Proof theory}
\ccsdesc[500]{Theory of computation~Type theory}
\ccsdesc[500]{Theory of computation~Categorical semantics}
\ccsdesc[500]{Theory of computation~Denotational semantics}

\keywords{Algebraic lambda calculus, Categorical semantics, Disjunction, Proof theory}

\category{} 

\funding{Supported by the European Union through the MSCA SE project QCOMICAL (Grant Agreement ID: 101182521), by the Plan France 2030 through the PEPR integrated project EPiQ (ANR-21-PETQ-0007), and by the Uruguayan CSIC grant 22520220100073UD.}

\nolinenumbers 

\EventEditors{John Q. Open and Joan R. Access}
\EventNoEds{2}
\EventLongTitle{45th IARCS Annual Conference on Foundations of Software Technology and Theoretical Computer Science}
\EventShortTitle{FSTTCS 2025}
\EventAcronym{FSTTCS}
\EventYear{2025}
\EventDate{December 17--19, 2025}
\EventLocation{Sancoale, Goa, India}
\EventLogo{}
\SeriesVolume{45}
\ArticleNo{}

\usepackage{xspace}
\usepackage{mathtools}
\usepackage{adjustbox}
\usepackage{needspace}
\usepackage{wrapfig}
\allowdisplaybreaks

\usepackage{proof}
\usepackage{cmll}
\usepackage{tikz}
\usetikzlibrary{cd,positioning,decorations.text}
\tikzcdset{scale cd/.style={every label/.append style={scale=#1},cells={nodes={scale=#1}}}}
\usepackage{stmaryrd}
\DeclareFontFamily{U}{stmry}{}
\DeclareFontShape{U}{stmry}{b}{n}
{  <4.25>stmary5
  <5> <6> <7> <8> <8.5> <9> <10> gen * stmary
  <10.95><12><14.4><17.28><20.74><24.88>stmary10%
}{}
\DeclareFontShape{U}{stmry}{m}{n}
{ <4.25>stmary5 
  <5> <6> <7> <8> <8.5> <9> <9.5> <10> gen * stmary
  <10.95><12><14.4><17.28><20.74><24.88>stmary10%
}{}
\DeclareFontFamily{U}{cmllr}{}
\DeclareFontShape{U}{cmllr}{bx}{n}
{ <4.25>stmary5 
  <5> <6> <7> <8> <8.5> <9> <9.5> <10> gen * stmary
  <10.95><12><14.4><17.28><20.74><24.88>stmary10%
}{}
\DeclareFontShape{U}{cmllr}{m}{n}
{ <4.25>stmary5 
  <5> <6> <7> <8> <8.5> <9> <9.5> <10> gen * stmary
  <10.95><12><14.4><17.28><20.74><24.88>stmary10%
}{}

\newcommand\Obj{\mathsf{Obj}}
\newcommand\Arr{\mathsf{Arr}}
\newcommand\restate[4]{\noindent{\textcolor{lipicsGray}{$\blacktriangleright$}\nobreakspace\sffamily\bfseries #1~\ref{#2}~{\normalfont\sffamily (#3)}.}\nobreakspace{\em #4}}
\newcommand\scal[1]{\mathfrak{#1}}
\newcommand\sstare[1]{\protect{\star_{#1}}}
\newcommand\sstar[1]{\sstare{\scal #1}}
\newcommand\irule[3]{\infer[\mbox{\footnotesize $#3$}]{#2}{#1}}
\newcommand\TheCat{\ensuremath{\mathbf{Mag}_{\mathbf{Set}}}\xspace}

\newcommand\TheCatAlg{\ensuremath{\mathbf{AMag}^{\mathcal{S}}_{\mathbf{Set}}}\xspace}
\newcommand\alglambda{\lambda_{\textsf{alg}}^{\mathcal S}}
\newcommand\plus{\parallel}
\newcommand\inl{\mathsf{inl}}
\newcommand\inr{\mathsf{inr}}
\newcommand\elimtop{\delta_{\top}}
\newcommand\elimbot{\delta_{\bot}}
\newcommand\elimor{\delta_{\vee}}
\newcommand\pair[2]{\langle #1, #2 \rangle}
\newcommand\abstr[1]{#1.}
\newcommand\xlra[1]{\xrightarrow{#1}}
\newcommand\home[2]{[#1\to #2]}
\newcommand\sem[1]{\left\llbracket {#1}\right\rrbracket}
\newcommand\Id{\mathsf{id}}
\newcommand\lra{\longrightarrow}
\newcommand\lla{\longleftarrow}
\newcommand\inclusion{\mathrm{inc}}
\newcommand\coproducto[2]{[#1,#2]}
\newcommand\cp{\ensuremath{\mathbin{\raisebox{0.3ex}{\scalebox{0.8}{\ooalign{$\bigcirc$\cr${\scalebox{0.8}{\hspace{0.05mm}$\dashv\vdash$}}$}}}}}}
\newcommand\cplabel{\scalebox{0.7}{$\cp$}}
\newcommand\symbolOpSG{\mathbin{\scalebox{1.1}{$\ast$}}}
\newcommand\sumhat[1]{\mathbin{\symbolOpSG_{{}_{\!#1}}}}
\newcommand\prodhat[1]{\mathbin{\bullet_{{}_{\!#1}}}}
\newcommand\subcp{\scalebox{0.5}{\cp}}
\newcommand\sumcp{\protect{\ensuremath{\sumhat{\subcp}}}}
\newcommand\prodcp{\ensuremath{\prodhat{\subcp}}}

\begin{document}

\maketitle


\begin{abstract}
  Traditional approaches to modelling parallelism and algebraic structure in lambda calculi often rely on monads---as in Moggi's framework---or on rich categorical structures such as biproducts---as used in certain models of linear logic. In this work, we propose a minimal alternative that captures both parallelism and weighted parallelism (linear combinations) within the setting of intuitionistic propositional logic, without resorting to monads or assuming the existence of biproducts.

We introduce two lambda calculi: a parallel lambda calculus and an algebraic lambda calculus, both extending full propositional intuitionistic logic. Their semantics are given in two categories: ${\mathbf{Mag}_{\mathbf{Set}}}$, whose objects are magmas and arrows are functions in $\mathbf{Set}$; and ${\mathbf{AMag}^{\mathcal{S}}_{\mathbf{Set}}}$, whose objects are action magmas.

The key technical challenge addressed is the interpretation of disjunction in the presence of parallel and algebraic operators. Since the usual coproduct structure is unavailable in our minimal setting, we propose a novel set-theoretic interpretation based on the union of the disjoint union and the Cartesian product. This allows for the construction of sound and adequate models for both calculi.

Our results offer a unified and structurally lightweight framework for modelling parallelism and algebraic effects in intuitionistic logic, opening the way to alternatives beyond the traditional monadic or linear logic approaches.
\end{abstract}

\section{Introduction}
\subparagraph*{Motivation}
In his 1991 paper~\cite{Moggi91}, Eugenio Moggi introduced a uniform framework to handle effectful computations in the lambda calculus, using monads to encapsulate various computational effects such as non-determinism, state, exceptions, and input/output. This monadic approach has since been widely adopted in functional programming languages, providing a powerful abstraction for reasoning about side effects. In the context of intuitionistic propositional logic---which corresponds categorically to cartesian closed categories---non-determinism can be captured, for instance, in the category $\mathbf{Set}$ by the map $\chi : A \times A \to \mathcal{P}A$ defined as $(x,y) \mapsto \{x,y\}$. This non-deterministic operator is sometimes referred to as a \emph{parallel operator}, as it can be seen as preserving all possible outputs in parallel.

While this monadic framework is highly versatile, it does not extend seamlessly
to all logical systems. In particular, within Intuitionistic Linear Logic
(ILL), the monadic encoding of non-determinism becomes problematic: the map
$\chi$ is not linear. For instance, $\chi(x_1 + x_2, y_1 + y_2) = \{x_1 + y_1,
x_2 + y_2\}$, whereas $\chi(x_1, y_1) + \chi(x_2, y_2) = \{x_1, y_1\} + \{x_2,
y_2\}$, which does not yield the same set under any reasonable definition of
set addition.

An alternative approach, suitable for linear settings, was developed through a
line of work on relational models. In particular,
in~\cite{BucciarelliEhrhardManzonettoLFCS09}, building on previous work
developed in Manzonetto’s doctoral thesis~\cite{ManzonettoThesis08} and in
\cite{BucciarelliEhrhardManzonettoCSL07}, and in the context of untyped
$\lambda$-calculus models, the hom-sets of the category $\mathbf{MRel}$,
restricted to a reflexive object, were enriched with a commutative semiring
structure in order to model a parallel and non-deterministic $\lambda$-calculus.
In the typed setting, this approach was pursued in weighted relational models
of typed $\lambda$-calculi~\cite{LairdManzonettoMcCuskerPaganiLICS13}, where the
effect is captured by the map $\chi' : A \times A \to A$ defined by
$(x,y) \mapsto x+y$. In~\cite{DiazcaroMalherbe24}, this idea was extended into a
more general categorical characterisation, relying on the presence of biproducts
to define the sum.

The guiding question of the present work is the following: \emph{Can we construct an alternative to Moggi's monad, inspired by the use of biproducts, that works in the context of propositional logic, where biproducts are not available?} We show that the answer is affirmative by providing a concrete model that captures this effect without requiring the full structure of a biproduct.

Our proposal is a structurally simple alternative capable of capturing both
parallelism and weighted parallelism (or linear combinations). We move from the
category $\mathbf{Set}$ to a richer setting based on magmas, where we define an
ad hoc operation that mimics---to some extent---the behaviour of a biproduct,
though it is not a biproduct in the categorical sense.

More precisely, we introduce the category $\TheCat$, whose objects are magmas with no required algebraic properties, and whose arrows are simply functions from the category $\mathbf{Set}$. In addition, we define the category $\TheCatAlg$, whose objects are \emph{action magmas}---magmas equipped with an external scalar action---and whose arrows are again functions from $\mathbf{Set}$. These constructions allow us to interpret the parallel lambda calculus and the algebraic lambda calculus, respectively, in a way reminiscent of the interpretation of the parallel linear lambda calculus in the category of vector spaces, but without requiring the existence of biproducts.

A key difficulty in adapting these techniques to intuitionistic propositional logic lies in the treatment of disjunction. In contrast to  models of linear logic equipped with biproducts---where additive conjunction and disjunction are both interpreted via this common structure---the absence of a unified product-coproduct structure in propositional logic means that disjunction cannot be handled in the same way as conjunction. This asymmetry poses significant challenges in modelling the behaviour of parallel or algebraic combinations involving disjunctive terms. In this work, we address this difficulty directly by proposing a concrete interpretation of disjunction that preserves the intuition of additive structure, while remaining compatible with the categorical limitations of the setting.

\subparagraph*{Antecedents}
Several foundational works have explored extensions of the lambda calculus to express superposition, linear combinations, and parallelism. The Lineal calculus~\cite{ArrighiDowekLMCS17} was introduced to represent quantum programs by incorporating the principle that, since data and programs are unified in the lambda calculus, the superposition of data implies the superposition of programs. This idea was formalised through a linear structure supporting algebraic combinations of terms.

An independent development is the Algebraic Lambda Calculus~\cite{Vaux2009}, a simplification of the Differential Lambda Calculus~\cite{EhrhardRegnierTCS03}, which emphasises algebraic structure rather than operational semantics. Later, Lineal and the Algebraic Lambda Calculus were shown to be closely related~\cite{AssafDiazcaroPerdrixTassonValironLMCS14}.

Several semantic models have addressed non-determinism and parallelism. In~\cite{DezanideLiguoroPipernoSIAM98}, a filter model was proposed for a concurrent lambda calculus, where terms are interpreted as sets of possible outcomes. A relational model capturing non-deterministic and parallel behaviour was introduced in~\cite{BucciarelliEhrhardManzonettoAPAL12}, and a general framework using weighted relations, allowing for deterministic, non-deterministic, and probabilistic computation, was proposed in~\cite{LairdManzonettoMcCuskerPaganiLICS13}.

More recently, the $\odot$-connective was introduced in~\cite{DiazcaroDowekTCS23a} and extended in~\cite{DiazcaroDowekMSCS24} to express quantum superpositions and measurements within a logical framework. These systems include additive structure in proof terms, such as sums and scalar products, aiming to represent quantum behaviour. A categorical interpretation for such calculi in a linear setting---using biproducts in the category of vector spaces---was developed in~\cite{DiazcaroMalherbe24}.

\subparagraph*{Our perspective}
While Moggi's monadic approach provides a general and robust framework for
handling computational effects in intuitionistic settings, our goal is to
explore whether the structural techniques developed in the linear logic
context---particularly those relying on biproducts---can be adapted to
intuitionistic propositional logic. Instead of using a monad to encapsulate
non-determinism or parallelism, we investigate how to reconstruct these effects
using categorical structures that simulate the role of biproducts, even in
their absence. In this way, we aim to propose a conceptually unified
alternative for modelling such effects across both linear and non-linear
logical frameworks.
\subparagraph*{Plan of the paper and contributions.}
In this paper, we introduce two lambda calculi for propositional logic—one modelling parallelism and the other algebraic combinations—based on minimal categorical structure inspired by biproducts, and avoiding the use of monads.
\begin{itemize}
  \item In \cref{sec:calculus}, we introduce the parallel lambda calculus $\lambda_\parallel$ for the full intuitionistic propositional logic. 
    It is mostly based on the in-left-right-+-calculus (inlr-calculus for short)~\cite{DiazcaroDowekInlr}, with a slight modification. In the inlr-calculus the parallel of $\inl(t)$ and $\inr(u)$ reduces to a new term $\mathsf{inlr}(t,u)$. In our presentation we just keep $\inl(t)\plus\inr(u)$, however, the two calculi are equivalent.
    This calculus follows the approach of~\cite{CirsteaFaureKirchnerHOSC07} for the structure operator, where the parallel operator is neither idempotent, commutative, nor associative.

      \cref{sec:syntax} gives the syntax, deduction, and reduction rules. 
      \cref{sec:correctness} restate the correctness properties of the calculus, proved in \cite{DiazcaroDowekInlr} and adapted to our setting in \appref{A}.
      \cref{sec:category} introduces the category \TheCat, which is the category whose objects are magmas and whose arrows are functions from the $\mathbf{Set}$ category. We also introduce the product $\cp$ and prove several necessary properties. 
      \cref{sec:model} provides the interpretation of $\lambda_\parallel$ within the category \TheCat, and prove its soundness and adequacy results, that is, if one term reduces to another, they are interpreted by the same arrow, and if two closed proof-terms are interpreted by the same arrows, then they are computationally equivalent.
  \item In \cref{sec:algebraiccalculus}, we introduce the algebraic lambda calculus $\alglambda$ for the full intuitionistic propositional logic. This calculus is an extension of the $\lambda_\parallel$ calculus with scalars and follows the approach of the Algebraic lambda calculus~\cite{Vaux2009}.

      \cref{sec:syntaxAlg} gives the syntax, deduction, and reduction rules.
      \cref{sec:correctnessalg} states the correctness properties of the calculus as straightforward adaptations of those from \cref{sec:correctness}.
      \cref{sec:categoryalg} introduces the concept of action magmas and define the category \TheCatAlg, whose objects are action magmas, and whose arrows are functions from the $\mathbf{Set}$ category. We also extend the product $\cp$ into an action magma. 
      \cref{sec:modelAlg} provides the interpretation of the algebraic lambda calculus within the category \TheCatAlg, and prove its soundness and adequacy.
  \item In \cref{sec:conclusion}, we present some concluding remarks.
\end{itemize}

\section{The \texorpdfstring{calculus $\lambda_\parallel$}{parallel lambda calculus}}\label{sec:calculus}
\subsection{Syntax, deduction, and reduction rules}\label{sec:syntax}

The connectives in this logic are those found in standard intuitionistic propositional logic: $\top$, $\bot$, $\Rightarrow$, $\wedge$, and $\vee$.
The syntax of the proof-terms are also standard, with the addition of the parallel construction as mentioned in the introduction.
\[
  \begin{array}{rl@{\qquad}l}
    t =~ & x \mid t \plus t \\ 
    & \mid \star \mid \elimtop(t,t) & (\top)\\
    &\mid \elimbot(t) & (\bot) \\
    & \mid \lambda \abstr{x}t\mid t~t & (\Rightarrow)\\
    & \mid \pair{t}{t} \mid \pi_1(t)\mid\pi_2(t) & (\wedge) \\
    & \mid \inl(t)\mid \inr(t) \mid \elimor(t,\abstr{x}t,\abstr{y}t) & (\vee)
  \end{array}
\]
where $x$ and $y$ range over a countable set of variables.

The terms $\star$, $\lambda x.t$, $\pair tu$, $\inl(t)$, and $\inr(t)$ are called introductions. The terms $\elimtop(t,u)$, $\elimbot(t)$, $\pi_1(t)$, $\pi_2(t)$, and $\elimor(t,\abstr{x}u,\abstr{y}v)$ are called eliminations. Variables and $t\plus u$ are neither introductions nor eliminations, except for $\inl(t)\plus\inr(u)$ which can be considered as an introduction. Free variables are defined as usual, and the substitution of $x$ by $u$ in $t$ is written $(u/x)t$, with the usual renaming to avoid variable captures. 

Our notation is inspired by logic. However, $\elimtop(t,r)$ could also be written $t;r$ to denote sequencing, $\elimbot(t)$ as $\mathsf{error}(t)$ to denote raising an exception, $\elimor(t,\abstr{x}u,\abstr{y}v)$ as $\mathsf{case}~t~\mathsf{of}~\{\inl(x) \to u; \inr(y) \to v\}$ for pattern matching, and $\pi_1(t)$ and $\pi_2(t)$ as $\mathsf{fst}(t)$ and $\mathsf{snd}(t)$ for projections. From this viewpoint, the calculus can be seen as a programming language.

The also standard deduction rules are the following: 
  $$\irule{}
    {x:A,\Gamma \vdash x:A}
    {\mbox{ax}}
    \quad
    \irule{\Gamma \vdash t:A & \Gamma \vdash u:A}
    {\Gamma \vdash t \plus u:A}
    {\mbox{par}}
    \quad
    \irule{}
    {\Gamma \vdash \star:\top}
    {\top_i}
    \quad
    \irule{\Gamma \vdash t:\top & \Gamma \vdash u:C}
    {\Gamma \vdash \elimtop(t,u):C}
    {\top_e}
  $$
  $$
    \irule{\Gamma \vdash t:\bot}
    {\Gamma \vdash \elimbot(t):C}
    {\bot_e}
    \qquad
    \irule{x:A,\Gamma \vdash t:B}
    {\Gamma \vdash \lambda \abstr{x}t:A \Rightarrow B}
    {\Rightarrow_i}
    \qquad
    \irule{\Gamma \vdash t:A\Rightarrow B & \Gamma \vdash u:A}
    {\Gamma \vdash t~u:B}
    {\Rightarrow_e}$$
  $$\irule{\Gamma \vdash t:A & \Gamma \vdash u:B}
    {\Gamma \vdash \pair{t}{u}:A \wedge B}
    {\wedge_i}
    \qquad
    \irule{\Gamma \vdash t:A \wedge B}
    {\Gamma \vdash \pi_1(t):A}
    {\wedge_{e1}}
    \qquad
    \irule{\Gamma \vdash t:A \wedge B}
    {\Gamma \vdash \pi_2(t):B}
    {\wedge_{e2}}$$
  $$\irule{\Gamma \vdash t:A}
    {\Gamma \vdash \inl(t):A \vee B}
    {\vee_{i1}}
    \qquad
    \irule{\Gamma \vdash t:B}
    {\Gamma \vdash \inr(t):A \vee B}
    {\vee_{i2}}
  $$
  $$
    \irule{\Gamma \vdash t:A \vee B & x:A,\Gamma \vdash u:C & y:B,\Gamma \vdash v:C}
    {\Gamma \vdash \elimor(t,\abstr{x}u,\abstr{y}v):C}
    {\vee_e}$$

The reduction relation is the contextual closure of the relation defined by the rules:
\parbox{0.35\textwidth}{
  \begin{align}
    \elimtop(\star, t) & \longrightarrow t \label{ruelimtop}\\
    (\lambda \abstr{x}t)~u & \longrightarrow  (u/x)t \label{rubeta}\\
    \pi_1\pair{t}{u} & \longrightarrow  t \label{ruelimand1}\\
    \pi_2\pair{t}{u} & \longrightarrow  u \label{ruelimand2}
  \end{align}
}\quad \parbox{0.62\textwidth}{
  \begin{align}
    \elimor(\inl(t),\abstr{x}v,\abstr{y}w) & \longrightarrow  (t/x)v \label{ruelimorinl}\\
    \elimor(\inr(u),\abstr{x}v,\abstr{y}w) & \longrightarrow (u/y)w \label{ruelimorinr} \\
    \elimor(\inl(t)\plus\inr(u),\abstr{x}v,\abstr{y}w) & \longrightarrow  (t/x)v \plus (u/y)w \label{ruelimorinlr1}
\end{align}}\vspace{-1\baselineskip}
\begin{align}
  {\star} \plus \star & \lra  \star \label{rusumstar} \\
  (\lambda \abstr{x}t) \plus (\lambda \abstr{x}u) & \lra  \lambda \abstr{x}(t \plus u) \label{rusumlam}\\
  \pair{t_1}{u_1} \plus \pair{t_2}{u_2}  & \lra  \pair{t_1 \plus t_2}{u_1 \plus u_2} \label{rusumpair} \\
  \inl(t_1)\plus\inl(t_2) & \lra \inl(t_1\plus t_2) \label{rusuminl}\\
  \inl(t_1)\plus(\inl(t_2)\plus\inr(u_1)) & \lra \inl(t_1\plus t_2)\plus\inr(u_1)\label{rusuminlinlr}\\
  \inr(u_1)\plus\inl(t_1) & \lra \inl(t_1)\plus\inr(u_1) \label{rusuminrinl}\\
  \inr(u_1)\plus\inr(u_2) & \lra \inr(u_1\plus u_2) \label{rusuminr}\\
  \inr(u_1)\plus(\inl(t_1)\plus\inr(u_2)) & \lra \inl(t_1)\plus\inr(u_1\plus u_2) \label{rusuminrinlr}\\
  (\inl(t_1)\plus\inr(u_1))\plus\inl(t_2) &\lra \inl(t_1\plus t_2)\plus\inr(u_1) \label{rusuminlrinl}\\
  (\inl(t_1)\plus\inr(u_1))\plus\inr(u_2) &\lra \inl(t_1)\plus\inr(u_1\plus u_2) \label{rusuminlrinr}\\
  (\inl(t_1)\plus\inr(u_1))\plus(\inl(t_2)\plus\inr(u_2)) &\lra \inl(t_1\plus t_2)\plus\inr(u_1\plus u_2) \label{rusuminlr}
\end{align}

The first group of reduction rules is standard, except for the last rule, which allows reducing $\inl$ and $\inr$ in parallel. The second group establishes the commutation of the parallel construction with all the connectives. The parallel construction commutes with introductions, as seen in the first three rules of the second group for the cases of $\top$, $\Rightarrow$, and $\wedge$. The case of $\vee$ is left to the next eight rules since multiple cases must be considered due to the fact that there are two normal forms for disjunctions, and disjunction is neither commutative nor associative in Natural Deduction.

Instead of commuting the parallel construction with the introduction of disjunction, we could have considered doing so with its elimination, using the rule $\elimor(t\plus u,x.v,y.w)\lra\elimor(t,x.u,y.w)\plus\elimor(u,x.v,y.w)$, which can also be proven valid in our model (the proof of soundness for this rule is given in \appref{C}). However, we chose to consider commutation with the introduction, as it leads to a better introduction property: a closed irreducible proof of $\top$ is $\star$, a closed irreducible proof of an implication is a lambda abstraction, a closed irreducible proof of a conjunction is a pair, and a closed irreducible proof of a disjunction is either $\inl$, $\inr$, or the parallel of both  (see \cref{thm:IP}).

Along this paper if $R$ is a relation, we write $R^*$ for its transitive-reflexive closure.

\subsection{Correctness}\label{sec:correctness}
The correctness properties of this section have been proven in a preprint for the inlr calculus~\cite{DiazcaroDowekInlr}, and the proofs are the same for the $\lambda_\parallel$ calculus.

\begin{theorem}
  [Subject reduction~{\cite[Appendix A.1]{DiazcaroDowekInlr}}]\label{thm:SR}
  If $\Gamma\vdash t:A$ and $t\lra u$, then $\Gamma\vdash u:A$.
  \qed
\end{theorem}

\begin{theorem}
  [Strong normalisation~{\cite[Appendix A.4]{DiazcaroDowekInlr}}]
  \label{thm:SN}
  If $\Gamma\vdash t:A$ then $t$ is strongly normalising.
  \qed
\end{theorem}

\begin{theorem}[Confluence]
  \label{thm:conf}
  The $\lambda_\parallel$ calculus is confluent.
\end{theorem}
\begin{proof}
  The relation $\lra$ from \cref{sec:syntax} is left linear and has no critical pairs. And, as proved in~\cite[Theorem 6.8]{Nipkow}, higher-order left linear systems without critical pairs are confluent.
\end{proof}

\begin{theorem}[Introduction property~{\cite[Appendix A.3]{DiazcaroDowekInlr}}]
  \label{thm:IP}
  Let $\vdash t:A$ be irreducible. Then: 
    If $A=\top$, $t=\star$.
    $A$ is not $\bot$.
    If $A=B \Rightarrow C$, $t=\lambda \abstr{x}u$.
    If $A=B \wedge C$, $t=\pair{u}{v}$.
    If $A=B \vee C$, $t=\inl(u)$, $t=\inr(v)$, or $t=\inl(u)\plus\inr(v)$.
      \qed
\end{theorem}

\subsection{The category \texorpdfstring{\TheCat}{Mag-Set}}\label{sec:category}
\begin{definition}[The category \TheCat]
  The category \TheCat is determined by the following data:
    Objects are magmas $(A,\sumhat A)$.
    Arrows are maps in $\mathbf{Set}$.
\end{definition}

  \begin{remark}
    Note that $\TheCat$ is not an algebraic category, in the sense that the natural arrows for the algebraic structure (such as homomorphisms) are not considered. Instead, the objects are plain magmas, without any algebraic property, and the arrows are simply functions in $\mathbf{Set}$, with composition and identities inherited from $\mathbf{Set}$. The reasons for this choice are discussed in detail in \cref{sec:conclusion}.
  \end{remark}

The next definitions give some objects and arrows used later.
\begin{definition}
  \label{def:objects}
  Let $A$ and $B$ be objects in \TheCat, we define the following objects.
  \begin{itemize}
    \item $\emptyset$ with the empty map as its operation.
    \item $\{\star\}$ with the operation defined as $\star\sumhat{\{\star\}}\star=\star$.
    \item $A\times B$ with the operation defined as
          $(a_1,b_1)\sumhat{A\times B}(a_2,b_2) = (a_1\sumhat A a_2,b_1\sumhat B b_2)$.
    \item $\home AB$ with the operation defined as
          $(f\sumhat{\home AB}g)(a) = f(a)\sumhat{B}g(a)$.
  \end{itemize}
\end{definition}

\begin{theorem}[CCC]
  \label{thm:closure}
  The category \TheCat is Cartesian closed.
\end{theorem}
\begin{proof}
  We need to prove that there is an adjunction $A\times\_\dashv\home{A}{\_}$ in \TheCat.
  For all $A,B$ objects in \TheCat, we have $Hom_{\TheCat}(A,B) = Hom_{\mathbf{Set}}(A,B)$.
  Since, $A\times B$ and $\home BC$ are objects in \TheCat,
  by the adjunction in $\mathbf{Set}$, we have
  \(
    Hom_{\TheCat}(A\times B,C)
    =
    Hom_{\mathbf{Set}}(A\times B,C) 
    \simeq
    Hom_{\mathbf{Set}}(A,\home BC)  
    =
    Hom_{\TheCat}(A,\home BC)
  \).
\end{proof}

As usual, the disjoint union of two sets $A$ and $B$ is denoted by $A\uplus B$ and is defined as $A\uplus B = (A\times\{0\})\cup (B\times\{1\})$ where $0$ and $1$ are distinct elements that do not belong to either to $A$ or $B$. If we want $A\uplus B$ to be an object in \TheCat, we must define its magma operation. For example, we can choose $(a_1,0)\sumhat{\uplus}(a_2,0) = (a_1\sumhat A a_2,0)$, and similarly $(b_1,1)\sumhat{\uplus}(b_2,1) = (b_1\sumhat B b_2,1)$. However, it is not clear how to define $(a,0)\sumhat{\uplus}(b,1)$. We could choose to define it as either $(a,0)$ or $(b,1)$, but this would break the symmetry. This is the same problem encountered in $\inl(t)\plus\inr(u)$. Indeed, in the sequent $\vdash\inl(t)\plus\inr(u):A\vee B$, we are given both a proof of $A$ and a proof of $B$. However, to produce a proof of $A\vee B$, we would typically need to discard one of them, whereas with the parallel operator, we can retain both. This is why we need to define a new operation, which we call pokeball product and denote by $\cp$.

\begin{definition}[The pokeball product]
  We define $A\cp B = (A\uplus B)\cup (A\times B)$.
\end{definition}
\begin{definition}
  [The operation {\normalfont $\sumcp$}]
  \label{def:sumcp}
  Let $(A,\sumhat A),(B,\sumhat B)\in\Obj(\TheCat)$.
  We define $\sumhat{A\subcp B}$ ($\sumcp$ for short) as 
  $(A\cp B)\times (A\cp B)\xlra{\sumcp} A\cp B$ given by 
  \[
    (c_1,c_2)\mapsto
    \left\{
    \begin{array}{ll}
      (a_1\sumhat A a_2,0)                & \mbox{if } c_1 = (a_1,0)\mbox{ and }c_2 =(a_2,0)     \\
      (a,b)                               & \mbox{if } c_1 = (a,0)\mbox{ and }c_2 =(b,1)         \\
      (a_1\sumhat A a_2,b)                & \mbox{if } c_1 = (a_1,0)\mbox{ and }c_2 =(a_2,b)     \\
      (a,b)                               & \mbox{if } c_1 = (b,1)\mbox{ and }c_2 =(a,0)         \\
      (b_1\sumhat B b_2,1)                & \mbox{if } c_1 = (b_1,1)\mbox{ and }c_2 =(b_2,1)     \\
      (a,b_1\sumhat B b_2)                & \mbox{if } c_1 = (b_1,1)\mbox{ and }c_2 =(a,b_2)     \\
      (a_1\sumhat A a_2,b)                & \mbox{if } c_1 = (a_1,b)\mbox{ and }c_2 =(a_2,0)     \\
      (a,b_1\sumhat B b_2)                & \mbox{if } c_1 = (a,b_1)\mbox{ and }c_2 =(b_2,1)     \\
      (a_1\sumhat A a_2,b_1\sumhat B b_2) & \mbox{if } c_1 = (a_1,b_1)\mbox{ and }c_2 =(a_2,b_2)
    \end{array}
    \right.
  \]
\end{definition}
\begin{definition}
  Let $A\xlra{f}A'$ and $B\xlra{g}B'$.
  The arrow $A\cp B\xlra{f\cplabel g} A'\cp B'$ is defined as follows.
  \[
    c\mapsto
    \left\{
    \begin{array}{ll}
      (f(a),0)    & \mbox{if } c = (a,0) \\
      (g(b),1)    & \mbox{if } c = (b,1) \\
      (f(a),g(b)) & \mbox{if } c = (a,b)
    \end{array}
    \right.
  \]
\end{definition}

The pokeball product entails a bifunctorial functor.
\begin{lemma}
  [Bifunctoriality]\label{lem:cpbifunctorial}
  Let $A\xlra{f}A'$, $A\xlra{f'}A'$, $B\xlra{g}B'$, and $B\xlra{g'}B'$.
  Then $A\cp\_$ is a functor, and it is bifunctorial, that is
  \(
    (f\cp g)\circ(f'\cp g') = (f\circ f')\cp(g\circ g')
  \).
\end{lemma}
\begin{proof}
  ~
  \begin{itemize}
    \item We prove that $A\cp\_$ is a functor. That is, we need to prove the commutation of the following diagrams.
      \begin{center}
	\begin{tikzcd}
	  A\cp B\ar[rd,"\Id\cplabel(h\circ g)"',sloped]\ar[r,"\Id\cplabel g"] & A\cp B'\ar[d,"\Id\cplabel h"]\\
	  & A\cp B'' 
	\end{tikzcd}
	\qquad\qquad
	\begin{tikzcd}
	  A\cp B\ar[r,"\Id_A\cplabel\Id_B"]\ar[r,"\Id_{A\subcp B}"'] & A\cp B
	\end{tikzcd}
      \end{center}
      We prove the diagrams by giving their evaluations. 
      \begin{center}
	\begin{tikzcd}
	  (a,0)\ar[mapsto,rd,"\Id\cplabel(h\circ g)"',sloped]\ar[mapsto,r,"\Id\cplabel g"] & (a,0)\ar[mapsto,d,"\Id\cplabel h"]\\
	  & (a,0)
	\end{tikzcd}
	\hfill
	\begin{tikzcd}
	  (b,1)\ar[mapsto,rd,"\Id\cplabel(h\circ g)"',sloped]\ar[mapsto,r,"\Id\cplabel g"] & (g(b),1)\ar[mapsto,d,"\Id\cplabel h"] \\
	  & (h(g(b)),1) 
	\end{tikzcd}
	\hfill
	\begin{tikzcd}
	  (a,b)\ar[mapsto,rd,"\Id\cplabel(h\circ g)"',sloped]\ar[mapsto,r,"\Id\cplabel g"] & (a,g(b))\ar[mapsto,d,"\Id\cplabel h"] \\
	  & (a,h(g(b))) 
	\end{tikzcd}
      \end{center}
      \begin{center}
	\begin{tikzcd}
	  (a,0)\ar[mapsto,r,"\Id_A\cplabel\Id_B"]\ar[mapsto,r,"\Id_{A\subcp B}"'] & (a,0)
	\end{tikzcd}
	\hfill
	\begin{tikzcd}
	  (b,1)\ar[mapsto,r,"\Id_A\cplabel\Id_B"]\ar[mapsto,r,"\Id_{A\subcp B}"'] & (b,1)
	\end{tikzcd}
	\hfill
	\begin{tikzcd}
	  (a,b)\ar[mapsto,r,"\Id_A\cplabel\Id_B"]\ar[mapsto,r,"\Id_{A\subcp B}"'] & (a,b)
	\end{tikzcd}
      \end{center}

      The proof of $\_\cp B$ is also analogous.
    \item To prove that it is bifunctorial it suffices to prove that the following diagram commutes~\cite[Proposition 3.1]{MacLane}.
      \begin{center}
	\begin{tikzcd}
	  A\cp B\ar[r,"\Id\cplabel g"]\ar[d,"f\cplabel \Id"'] & A\cp B'\ar[d,"f\cplabel  \Id"]\\
	  A'\cp B\ar[r,"\Id\cplabel g"'] & A'\cp B'
	\end{tikzcd}
      \end{center}

      We prove this diagram by giving its possible evaluations. 
      \[
	\begin{tikzcd}
	  (a,0)\ar[mapsto,r,"\Id\cplabel g"]\ar[mapsto,d,"f\cplabel \Id"'] & (a,0)\ar[mapsto,d,"f\cplabel  \Id"]\\
	  (f(a),0)\ar[mapsto,r,"\Id\cplabel g"'] & (f(a),0)
	\end{tikzcd}\qquad
	\begin{tikzcd}
	  (b,1)\ar[mapsto,r,"\Id\cplabel g"]\ar[mapsto,d,"f\cplabel \Id"'] & (g(b),1)\ar[mapsto,d,"f\cplabel  \Id"]\\
	  (b,1)\ar[mapsto,r,"\Id\cplabel g"'] & (g(b),1)
	\end{tikzcd}\qquad
	\begin{tikzcd}
	  (a,b)\ar[mapsto,r,"\Id\cplabel g"]\ar[mapsto,d,"f\cplabel \Id"'] & (a,g(b))\ar[mapsto,d,"f\cplabel  \Id"]\\
	  (f(a),b)\ar[mapsto,r,"\Id\cplabel g"'] & (f(a),g(b))
	\end{tikzcd}
	\tag*{\qedhere}
      \]
  \end{itemize}
\end{proof}

While the pokeball product will not be a coproduct in \TheCat (see Remark~\ref{rmk:coprodOnly}), we can define some maps that will allow us to interpret the disjunction with the pokeball product in a manner similar to a coproduct. First, we note that the usual injections \(i_1\) and \(i_2\) are magma homomorphisms.

\begin{lemma}\label{lem:injhomomorphism}
  The maps $A\xlra{i_1} A\cp B$ and $B\xlra{i_2} A\cp B$ are homomorphisms.
\end{lemma}
\begin{proof}
  ~
  \begin{itemize}
    \item
          $
            i_1(a_1)\sumcp i_1(a_2)  = (a_1,0)\sumcp(a_2,0)
            = (a_1\sumhat A a_2,0)
            = i_1(a_1\sumhat A a_2)
          $
    \item
          $
            i_2(b_1)\sumcp i_2(b_2)  = (b_1,1)\sumcp(b_2,1)
            = (b_1\sumhat B b_2,1)
            = i_2(b_1\sumhat B b_2)
          $
          \qedhere
  \end{itemize}
\end{proof}

We can define  the inclusion $A\times B\xlra{\inclusion} A\cp B$ by $(a,b)\mapsto (a,b)$.
An interesting property is that after two injections and the operation $\sumcp$, we can recover the original element, as stated by the following lemma.
\begin{lemma}
  \label{lem:inclusion}
  $\sumcp\circ(i_1\times i_2) = \inclusion$.
\end{lemma}
\begin{proof}
  Let $a\in A$ and $b\in B$. Then 
  $\sumcp\circ(i_1\times i_2)(a,b) = (a,0)\sumcp(b,1) = (a,b)$.
\end{proof}

\begin{definition}\label{def:coproduct}
  Let $A\xlra f C$ and $B\xlra g C$ be arrows in \TheCat.
  We define the arrow $A\cp B\xlra{\coproducto{f}{g}}C$ as follows.
  \[
    c\mapsto
    \left\{
    \begin{array}{ll}
      f(a)               & \mbox{if } c = (a,0) \\
      g(b)               & \mbox{if } c = (b,1) \\
      f(a)\sumhat C g(b) & \mbox{if } c = (a,b)
    \end{array}
    \right.
  \]
\end{definition}

\begin{lemma}[Weak coproduct]
  \label{lem:coprod}
  Then map $\coproducto{f}{g}$ makes the following diagram commute
  \[
    \begin{tikzcd}[column sep=2cm]
      A\ar[dr,"f"',sloped]\ar[r,"i_1"] & A\cp B\ar[d,"\coproducto{f}{g}",pos=0.4] & B\ar[dl,"g"',sloped]\ar[l,"i_2"']\\
      & C
    \end{tikzcd}
  \]
\end{lemma}
\begin{proof}
  ~
  \begin{itemize}
    \item Left triangle: $\coproducto{f}{g}(i_1(a)) = \coproducto{f}{g}(a,0) = f(a)$.
    \item Right triangle: $\coproducto{f}{g}(i_2(b)) = \coproducto{f}{g}(b,1) = g(b)$.
      \qedhere
  \end{itemize}
\end{proof}

\begin{remark}[Not a coproduct]
  \label{rmk:coprodOnly}
  \cref{lem:coprod} does not imply that $A\cp B$ is a coproduct in \TheCat, since the map $\coproducto{f}{g}$ is clearly not unique.

  However, in the category $\mathbf{Mag}$ of magmas, whose arrows are homomorphism (which is not Cartesian closed, and so it is not a good candidate for us), the map $\coproducto fg$ is indeed a coproduct, since it is easy to prove that it is an homomorphism, and its unicity can be proven using the homomorphism property as follows.

  Let $h:A\cp B\to C\in\Arr(\mathbf{Mag})$ such that
  \(
    h(i_1(a))  = h(a,0)  = f(a) 
  \) and
  \(
    h(i_2(b))  = h(b,1)  = g(b)
  \).
  Then,
  \(
    h(a,b) = h((a,0) \sumcp(b,1)) = h(a,0)\sumhat C h(b,1) = f(a)\sumhat C g(b)
  \).
  Therefore, if arrows are homomorphisms, $h = \coproducto{f}{g}$.
\end{remark}

\begin{definition}
  \label{def:arrows}
  Let $A$ and $B$ be objects in \TheCat, we define the following arrows.
  \begin{itemize}
    \item $A\xlra{!}\{\star\}$ defined by $a\mapsto \star$.
    \item $A\xlra\Delta A\times A$ defined by $a\mapsto (a,a)$.
    \item $(A\cp B)\times C\xlra d (A\times C)\cp (B\times C)$ defined by
      \[
	(e,c)\mapsto\left\{
	\begin{array}{ll}
	  ((a,c),0)  & \text{if } e = (a,0)\\
	  ((b,c),1)  & \text{if } e = (b,1)\\
	  ((a,c),(b,c)) & \text{if } e = (a,b)
	\end{array}
      \right.
      \]
    \item $A\times A\xlra{\symbolOpSG}  A$ defined by $(a_1,a_2)\mapsto a_1\sumhat Aa_2$.
    \item $\delta = (\Id\times\sigma\times\Id)\circ(\Id\times\Delta)$. That is, $A\times B\times C\xlra\delta A\times C\times B\times C$.
  \end{itemize}
\end{definition}

\subsection{Model of the \texorpdfstring{$\lambda_\parallel$ calculus}{parallel lambda calculus}}\label{sec:model}
\subsubsection{Interpretation}\label{sec:interpretation}
The interpretation of the propositions and contexts of the proof system is the following: 
  \[
    \begin{array}{rclr@{\qquad}|@{\qquad}lrcl}
      \multicolumn{3}{c}{\text{Interpretation of propositions}} &  &  & \multicolumn{3}{c}{\text{Interpretation of contexts}} \\
      \sem\top             & = &  \{\star\}            &&& \sem\emptyset    & = & \{\star\}                \\
      \sem\bot             & = & \emptyset             &&& \sem{x:A,\Gamma} & = & \sem A\times\sem{\Gamma} \\
      \sem{A\Rightarrow B} & = & \home{\sem A}{\sem B} &&& \\
      \sem{A\wedge B}      & = & \sem A\times\sem B   &&& \\
      \sem{A\vee B}        & = & \sem A\cp\sem B &&&
    \end{array}
  \]
  The interpretation of the deduction rules is the following 
  We refer to each deduction rule by its last sequent, since the system is syntax directed.
  \begin{align*}
    \sem{x:A,\Gamma \vdash x:A}
    & =\sem A\times\sem\Gamma\xlra{\pi_1}\sem A
    \\
    \sem{\Gamma \vdash t \plus u:A}
    & =\sem\Gamma\xlra{\Delta}\sem\Gamma\times\sem\Gamma\xlra{t\times u} \sem A\times\sem A\xlra{\sumhat{\sem A}}\sem A
    \\
    \sem{\Gamma \vdash \star:\top}
    & =\sem\Gamma\xlra{!}\{\star\}=\sem\top
    \\
    \sem{\Gamma \vdash \elimtop(t,u):C}
    & = \sem\Gamma\xlra{\Delta}\sem\Gamma\times\sem\Gamma\xlra{t\times u}\{\star\}\times\sem C
    \xlra{\pi_2}\sem C
    \\
    \sem{\Gamma \vdash \elimbot(t):C}
    & =\sem\Gamma\xlra t \emptyset\xlra{\emptyset}\sem C
    \\
    \sem{\Gamma \vdash \lambda \abstr{x}t:A \Rightarrow B}
    & =\sem\Gamma\xlra{\eta^{\sem A}}\home{\sem A}{\sem A\times \sem\Gamma}\xlra{\home{\sem A}{t}}\home{\sem A}{\sem B}
    \\
    \sem{\Gamma \vdash t~u:B}
    & = \sem\Gamma\xlra\Delta\sem\Gamma\times\sem\Gamma\xlra{t\times u}\home{\sem A}{\sem B}\times\sem A \xlra{\varepsilon}\sem B
    \\
    \sem{\Gamma \vdash \pair{t}{u}:A \wedge B}
    & = \sem\Gamma\xlra\Delta \sem\Gamma\times\sem\Gamma \xlra{t\times u} \sem A\times\sem B
    \\
    \sem{\Gamma \vdash \pi_1(t):A}
    & = \sem\Gamma\xlra{t}\sem{A\wedge B}=\sem A\times\sem B \xlra{\pi_1}\sem A
    \\
    \sem{\Gamma \vdash \pi_2(t):A}
    & = \sem\Gamma\xlra{t}\sem{A\wedge B}=\sem A\times\sem B \xlra{\pi_2}\sem B
    \\
    \sem{\Gamma \vdash \inl(t):A \vee B}
    & =\sem\Gamma\xlra{t}\sem A\xlra{i_1}\sem A\cp\sem B
    \\
    \sem{\Gamma \vdash \inr(t):A \vee B}
    & =\sem\Gamma\xlra{t}\sem B\xlra{i_2}\sem A\cp\sem B
    \\
    \sem{\Gamma \vdash \elimor(t,\abstr{x}u,\abstr{y}v):C}
    & = \sem{\Gamma}\xlra{\Delta}\sem{\Gamma}\times\sem{\Gamma} \xlra{t\times\Id} (\sem A\cp\sem B)\times\sem\Gamma\\
    &\phantom{=~} \xlra{d} (\sem A\times\sem\Gamma)\cp (\sem B\times\sem\Gamma) \xlra{\coproducto uv} \sem C
  \end{align*}

In some diagrams we may write $A$ instead of $\sem A$, to simplify the notation.
\subsubsection{Soundness}\label{sec:soundness}
The soundness of the model is proven in \cref{thm:soundness}. 
As usual, we require the following substitution lemma. Its proof is done by induction on $t$ and given in \appref{B.1}.

\begin{lemma}
  [Substitution]\label{lem:substitution}
  If ${x:A},\Gamma\vdash {t:B}$ and $\Gamma\vdash {u:A}$, then
  $\sem{\Gamma\vdash (u/x)t:B} = \sem{\Gamma\vdash t:B}\circ(\sem{\Gamma\vdash u:A}\times\Id)\circ\Delta$.
  \qed
\end{lemma}

\begin{theorem}
  [Soundness]\label{thm:soundness}
  If $t\lra r$ and $\Gamma\vdash t:A$, then
  $\sem{\Gamma\vdash t:A} = \sem{\Gamma\vdash r:A}$.
\end{theorem}
\begin{proof}
  By induction on the relation $\lra$.
  The full proof is given in \appref{B.2}.
  We only show two interesting cases here, involving the parallel operator.
  \begin{itemize}
    \item Rule $\elimor(\inl(t)\plus\inr(u),\abstr{x}v,\abstr{y}w)\longrightarrow (t/x)v\plus (u/y)w$. 
      The commuting diagram is
      \[
	\begin{tikzcd}[column sep=1cm,row sep=5mm]
	  {\Gamma\times\Gamma} &[-5mm] &[-4mm]&[-1.4cm] {(\Gamma\times\Gamma)\times\Gamma} \\
	  & {(\Gamma\times\Gamma)\times(\Gamma\times\Gamma)} \\
	  & {(\Gamma\times\Gamma)\times(\Gamma\times\Gamma)} \\
	  && {(A\times B)\times(\Gamma\times\Gamma)} \\
	  {C\times C} & {(A\times\Gamma)\times (B\times\Gamma)} && {(A\times B)\times\Gamma} \\
	  C &&& {(A\cp B)\times(A\cp B)\times\Gamma} \\
	  {(A\times\Gamma)\cp(B\times\Gamma)} &&& {(A\cp B)\times\Gamma}
	  \arrow["{\Delta\times\Id}", from=1-1, to=1-4]
	  \arrow["{\Delta\times\Delta}", dashed, from=1-1, to=2-2,sloped]
	  \arrow["{(t/x)v\times (u/x)w}"',sloped, from=1-1, to=5-1]
	  \arrow["{\Id\times\Delta}", dashed, from=1-4, to=2-2,sloped]
	  \arrow["{(t\times u)\times\Id}", from=1-4, to=5-4]
	  \arrow["{\Id\times\sigma\times\Id}", dashed, from=2-2, to=3-2]
	  \arrow["{(t\times\Id)\times (u\times\Id)}"', out=200,in=160, dashed, from=2-2, to=5-2]
	  \arrow["{(t\times\Id)\times (u\times\Id)}", dashed, from=3-2, to=5-2, pos=0.3]
	  \arrow["{\Id\times\sigma\times\Id}", dashed, from=4-3, to=5-2,sloped]
	  \arrow["{\sumhat C}"', from=5-1, to=6-1]
	  \arrow["{v\times w}", dashed, from=5-2, to=5-1]
	  \arrow["{\Id\times\Delta}", dashed, from=5-4, to=4-3,sloped]
	  \arrow["\delta", dashed, from=5-4, to=5-2]
	  \arrow["{(i_1\times i_2)\times\Id}", from=5-4, to=6-4]
	  \arrow["{\sumcp\times\Id}", from=6-4, to=7-4]
	  \arrow["{\coproducto{v}{w}}", from=7-1, to=6-1]
	  \arrow["d", from=7-4, to=7-1]
	  \arrow["\inclusion", from=5-2, to=7-1,sloped,dashed]
	  \arrow["\inclusion"', from=5-4, to=7-4,dashed,out=185,in=175,looseness=1.4]
	\end{tikzcd}
      \]
	The top diagram commutes by the functoriality of $\times$.
	The upper-left diagram commutes by \cref{lem:substitution} and the functoriality of $\times$.
	The upper-middle diagram commutes by coherence, when precomposed with $\Delta$s.
	The upper-right-top diagram commutes by the naturality of $(\Id\times\sigma\times\Id)\circ(\Id\times\Delta)$.
	The upper-right-bottom diagram is the definition of the map $\delta$.
	The bottom-left diagram commutes by definition of $\coproducto vw$.
	The bottom-right diagram commutes by definition of $\sumcp$.
	The bottom-middle diagram commutes since $\delta$ is a restriction of $d$.

    \item Rule $\inr(u_1) \plus (\inl(t_1)\plus\inr(u_2))  \longrightarrow \inl(t_1)\plus\inr(u_1 \plus u_2) $. The commuting diagram is
      \[
	\begin{tikzcd}[column sep=1cm]
	  {\Gamma\times(\Gamma\times\Gamma)} &{B\times (A\times B)}& {(A\cp B)\times((A\cp B)\times (A\cp B))} \\
	  {A\times(B\times B)} && {(A\cp B)\times (A\cp B)} \\
	  {A\times B} &{(A\cp B)\times(A\cp B)}& {A\cp B}
	  \arrow["{u_1\times (t_1\times u_2)}",sloped, from=1-1, to=1-2]
	  \arrow["{t_1\times(u_1\times u_2)}"', from=1-1, to=2-1]
	  \arrow["{\Id\times\sumcp}", from=1-3, to=2-3]
	  \arrow["{i_2\times(i_1\times i_2)}",sloped, from=1-2, to=1-3]
	  \arrow["{i_2\times\inclusion}"',sloped, dashed, from=1-2, to=2-3]
	  \arrow["{\simeq}"',sloped, dashed, from=2-1, to=1-2]
	  \arrow["{\Id\times\sumhat B}"', from=2-1, to=3-1]
	  \arrow["\sumcp", from=2-3, to=3-3]
	  \arrow["\inclusion", dashed, from=3-1, to=3-3,bend left=15,sloped]
	  \arrow["{i_1\times i_2}"', from=3-1, to=3-2,sloped]
	  \arrow["\sumcp"', from=3-2, to=3-3,sloped]
	\end{tikzcd}
      \]
	The upper-left diagram commutes by coherence.
	The upper-right diagram commutes by the functoriality of $\times$ and \cref{lem:inclusion}.
	The bottom diagram commutes by \cref{lem:inclusion}.
	The middle diagram commutes since
	    $\sumcp\circ (i_2\times\inclusion) (b_1,(a,b_2))
	    = (b_1,1)\sumcp (a,b_2)
	    = (a,b_1\sumhat B b_2)
	    $.
	    \qedhere
  \end{itemize}
\end{proof}

\needspace{2\baselineskip}
\begin{remark}
  ~
  \begin{itemize}
    \item The case of the rule
      $\elimor(\inl(t)\plus\inr(u),\abstr{x}v,\abstr{y}w)\longrightarrow
      (t/x)v\plus (u/y)w$ uses the fact that $\delta$ is a restriction of $d$.
      Indeed, $A\cp B$ includes $A\times B$, and thus, $\inl(t)\plus\inr(u)$
      can be seen as a pair of terms.

    \item The case of the rule $\inr(u_1) \plus (\inl(t_1)\plus\inr(u_2))
      \longrightarrow \inl(t_1)\plus\inr(u_1 \plus u_2)$ is solved by the
      definition of $\sumcp$. Indeed, we do not ask all the operations to be
      commutative or associative, but we can still use an operation with these
      properties to solve the problem.
  \end{itemize}
\end{remark}

\subsubsection{Adequacy}\label{sec:adequacy}
A proof of completeness would mean that if two closed proofs have the same interpretation, they both reduce to the same irreducible proof. However, this is too much to ask for in a calculus without $\eta$-reduction. Indeed, $\vdash \lambda x. tx: A \Rightarrow B$ and $\vdash t: A \Rightarrow B$ have the same interpretation, but they do not reduce to the same proof. We can only ask for the following weaker property, which is usually called adequacy. That is, if two proof-terms have the same interpretation, then they are computationally equivalent. The notion of computational equivalence is that no other program can distinguish between them. That is, putting them in any context will yield the same result. It is sufficient, however, to consider only elimination contexts of some proposition large enough to distinguish between two programs, such as $\top \vee \top$, the usual booleans.

Indeed, if $K$ is one of these contexts, putting two terms $t_1$ and $t_2$ inside $K$ is sufficient, and testing whether an introduction term such as $\pair Ku$ behaves the same when $t_1$ or $t_2$ are put inside $K$ does not add any additional information. Only $K$, the elimination context, is sufficient to distinguish between the two terms.

The notion of an elimination context is given below. Without loss of generality, we consider the elimination context to always produce proof-terms for smaller propositions.
\begin{definition}
  [Elimination context]\label{def:elimcontext}
  An elimination context is a proof-term produced by the following grammar, where $[\cdot]$ denotes a distinguished variable.
  \[
    K := [\cdot] \mid K~t\mid \pi_1(K)\mid \pi_2(K)\mid \elimor(K,\abstr xt,\abstr yu)
  \]
  where in the proof $\elimor(K,\abstr xt,\abstr yu)$, if $[\cdot]:C\vdash K:A\vee B$, then $x:A\vdash t:D$ and $y:B\vdash u:D$, with $|D|<|A\vee B|$.
  As usual, we write $K[t]$ for $(t/[\cdot])K$.
\end{definition}

We define the computational equivalence as follows.

\begin{definition}
  [Computational equivalence]
  Two proofs $\vdash t:A$ and $\vdash u:A$
  are computationally equivalent, written $t\sim u$,
  if for each elimination context $[\cdot]:A\vdash K:\top\vee\top$, there exists an irreducible proof $\vdash v:\top\vee\top$ such that $K[t]\lra^* v$ and $K[u]\lra^* v$.
\end{definition}

Finally, we can state and prove the adequacy of the model as follows.
\begin{theorem}
  [Adequacy]\label{thm:adequacy}
  If $\sem{\vdash t:A} = \sem{\vdash u:A}$, then $t\sim u$.
\end{theorem}
\begin{proof}
  By induction on the structure of $A$.
  The full details are given in \appref{B.3}.
  We only show one interesting case here.
  If $A=B\vee C$, then, by \cref{thm:SR,thm:SN,thm:IP},
  $t\lra^*\inl(t')$,
  $t\lra^*\inr(t')$,
  or
  $t\lra^*\inl(t')\plus\inr(t'')$;
  and
  $u\lra^*\inl(u')$,
  $u\lra^*\inr(u')$,
  $u\lra^*\inl(u')\plus\inr(u'')$.

  However, by \cref{thm:soundness}, and the fact that
  $\sem{\vdash t:A}=\sem{\vdash u:A}$, we have that
    if $t\lra^*\inl(t')$ then $u\lra^*\inl(u')$, 
    if $t\lra^*\inr(t')$ then $u\lra^*\inr(u')$, and 
    if $t\lra^*\inl(t')\plus\inr(t'')$ then $u\lra^*\inl(u')\plus\inr(u'')$.
  We must consider these three cases. We show here only the third case as example.
  We have
    $$\begin{tikzcd}[column sep=1cm,row sep=1cm]
      \{\star\}\ar[r,"\Delta"] & \{\star\}\times\{\star\}\ar[r,"t'\times t''",bend left]\ar[r,"u'\times u''"',bend right] & A\times B \ar[r,"i_1\times i_2"] &
      (A\cp B)\times (A\cp B)\ar[r,"\sumhat\times"]
      &
      A\cp B
    \end{tikzcd}$$
  Thus, $\sem{\vdash t':A} = \sem{\vdash u':A}$ and $\sem{\vdash t'':B} = \sem{\vdash u'':B}$, and so, by the induction hypothesis $t'\sim u'$ and $t''\sim u''$. Hence, $t\sim u$.
\end{proof}

\section{\texorpdfstring{The calculus $\alglambda$}{The algebraic lambda calculus}}\label{sec:algebraiccalculus}
\subsection{Syntax, deduction, and reduction rules}\label{sec:syntaxAlg}
In this section, we consider the $\alglambda$ calculus as an extension with scalars of the $\lambda_\parallel$ calculus. This calculus is similar to the Algebraic lambda calculus~\cite{Vaux2009}, but considering all the logical connectives, making it closer to the $\odot^{\mathcal S}$-calculus~\cite{DiazcaroDowekTCS23a}. The idea is to consider linear combinations of terms, as in $\scal a\cdot t\parallel\scal b\cdot u$, where we continue writing $\parallel$ for the operation instead of the more common $+$ from algebraic calculi, to emphasise its connection with the language defined in the previous section.

In the Algebraic lambda calculus~\cite{Vaux2009}, the scalars belong to a semiring with zero and unity (i.e.~a rig) and it is claimed not to use a ring since there are some confluence problems with negative scalars in its untyped setting. Indeed, a term $Y_t=(\lambda x.xx\plus t)(\lambda x.xx\plus t)$ reducing to $Y_t\plus t$ is the source of non-confluence when considering $\scal a\cdot Y_t\plus -\scal a\cdot Y_t$, which may reduce to $t\plus \cdots\plus t$ for any amount of $t$ in parallel (even zero), since at any point $\scal a\cdot Y_t$ may be cancelled with $-\scal a\cdot Y_t$ on these calculi.
There are several solutions to this problem. One is to consider only positive scalars, as done in the Algebraic lambda calculus. Another one is to consider a call-by-value beta reduction together with a linearity rule, as done in~\cite{ArrighiDowekLMCS17} (specifically, the solution involves restricting some rewriting rules to closed normal forms, which forces a call-by-value strategy). However, the simpler solution to this lack of confluence is to consider a typed setting ensuring strong normalisation, where the term $Y_t$ is not typable. For an extensive discussion on the different design choices for algebraic calculi, see~\cite{AssafDiazcaroPerdrixTassonValironLMCS14}. In our case, we relax all the restrictions on scalars and just require them to belong to a bi-magma, as defined next.
\begin{definition}[Bi-magma]
  A bi-magma is a set with two operations $(\mathcal S,+,\cdot)$.
\end{definition}

The syntax of the $\alglambda$ calculus extends the $\lambda_\parallel$ calculus introduced in \cref{sec:syntax} by incorporating scalars from a fixed bi-magma $\mathcal S$, as follows:
The term $\star$ is replaced with $\sstar s$, and a new term $\scal s \bullet t$ is added, where $\scal s \in \mathcal S$ and $t$ is a term.

The deduction rules are also those given in \cref{sec:syntax}, with the exception of rule $\top_i$, and the addition of new families of rules ``$\top_i(\scal s)$'' and ``prod($\scal s$)'' as follows.
\[
  \infer[\top_i(\scal s)]{\Gamma\vdash \sstar s:\top}{}
  \qquad
  \qquad
  \infer[\mbox{prod($\scal s$)}]{\Gamma\vdash \scal s\bullet t:A}{\vdash t:A}
\]
As with the rule par, the rules prod($\scal s$) do not introduce new provable formulas. They can be read as ``If we have a proof of $A$, then we have a proof of $A$''.

The reduction relation is similar to that of the $\lambda_\parallel$ calculus given in \cref{sec:syntax}, 
with modifications to the elimination of $\top$ and the commutation of $\plus$ with $\top_i(\scal s)$, since $\top_i$ now includes scalars:
  \begin{align*}
    \elimtop(\sstar s,t)                               & \longrightarrow  \scal s\bullet t &
    {\sstare{\scal s_1}} \plus {\sstare{\scal s_2}}                             & \longrightarrow  \sstare{\scal s_1+\scal s_2} 
  \end{align*}
  Additionally, a third group is added for the commutation rules of $\text{prod}(\scal s)$ with all the connectives. Indeed, just as with the rule $\text{par}$, the rule $\text{prod}(\scal s)$ can appear between the introduction and the elimination of a connective.
As in the case of the $\text{par}$ rule, we must commute it with either the elimination or the introduction. Following the same convention, we always commute the $\text{prod}(\scal s)$ rule with introductions:
  \begin{align*}
    \scal s_1\bullet \sstare{\scal s_2} & \longrightarrow  \sstare{\scal s_1\cdot\scal s_2}
    &\scal s\bullet\inl(t) & \longrightarrow \inl(\scal s\bullet t) 
    \\
    \scal s\bullet \lambda \abstr{x}t &\longrightarrow  \lambda \abstr{x}\scal s\bullet t 
    &\scal s\bullet\inr(t) & \longrightarrow \inr(\scal s\bullet t)  
    \\
    \scal s\bullet\pair{t}{v} & \longrightarrow  \pair{\scal s\bullet t}{\scal s\bullet v} 
    & \scal s\bullet(\inl(t)\plus\inr(v)) & \longrightarrow \inl(\scal s\bullet t)\plus\inr(\scal s\bullet v)
  \end{align*}

\subsection{Correctness}\label{sec:correctnessalg}
The correctness properties are straightforward adaptations of the proofs of \cref{thm:SR,thm:SN,thm:conf,thm:IP}, and are therefore omitted.
The only significant difference is that the introduction property now states that if $\vdash t : \top$ is irreducible, then $t = \sstar s$ for some $\scal s \in \mathcal S$.
%
%
%

\subsection{The category \texorpdfstring{\TheCatAlg}{AMag-S-Set}}\label{sec:categoryalg}
\begin{definition}[Action magma]
  An action magma $(A,\sumhat A,\prodhat A)$ over $\mathcal S$ consists of a magma $(A,\sumhat A)$ and a bi-magma $(\mathcal S,+,\cdot)$, with a map $\mathcal S\times A\xlra{\prodhat A}A$ such that it distributes with respect to $\sumhat A$, that is
  \(
    \scal s\prodhat A (a_1\sumhat A a_2) = \scal s\prodhat A a_1\sumhat A \scal s\prodhat A a_2
  \).
\end{definition}

\begin{definition}[The category \TheCatAlg]
  The category \TheCatAlg is determined by the following data:
    Objects are action magmas over the fixed bi-magma $\mathcal S$.
    Arrows are maps in $\mathbf{Set}$.
\end{definition}

\begin{theorem}[CCC]
  \label{thm:closureAlg}
  The category \TheCatAlg is Cartesian closed.
\end{theorem}
\begin{proof}
  The argument is the same as the one given in \cref{thm:closure}.
\end{proof}

We now extend the pokeball to the Action Magmas by defining its action.
\begin{definition}
  [The operation {\normalfont$\prodcp$}]\label{def:action}
  Let $(A,\sumhat A,\prodhat A),(B,\sumhat B,\prodhat B)\in\Obj(\TheCatAlg)$.
  We define $\prodhat{A\subcp B}$ ($\prodcp$ for short) as 
  $\mathcal S\times (A\cp B)\xlra{\prodcp} A\cp B$ 
  given by
  \[
    (\scal s,c)\mapsto
    \left\{
      \begin{array}{ll}
	(\scal s\prodhat A a,0)               & \mbox{if } c = (a,0) \\
	(\scal s\prodhat B b,1) 		& \mbox{if } c = (b,1) \\
	(\scal s\prodhat A a,\scal s\prodhat B b)   & \mbox{if } c = (a,b)
      \end{array}
    \right.
  \]
\end{definition}

\begin{definition}
  \label{def:objectsalg}
  Let $A,B\in\Obj(\TheCatAlg)$.
  The following are also objects in \TheCatAlg.
  \begin{itemize}
    \item $(A\cp B,\sumcp,\prodcp)$ with the action given in \cref{def:action}.
    \item $(\emptyset,\sumhat\emptyset,\prodhat\emptyset)$ with the action defined as the empty map.
    \item $(\{\star\},\sumhat{\{\star\}},\prodhat{\{\star\}})$ with the action defined as $\scal s\prodhat{\{\star\}}\star = \star$.
    \item A bi-magma $(\mathcal S,+,\cdot)$ with the magma operation defined by $+$ and the action defined by $\cdot$.
    \item $(A\times B,\sumhat{A\times B},\prodhat{A\times B})$ with the action defined by $\scal s\prodhat{A\times B}(a,b) = (\scal s\prodhat A a,\scal s\prodhat B b)$.
    \item $(\home AB,\sumhat{\home AB},\prodhat{\home AB})$ with the action defined by $(\scal s\prodhat{\home AB}f)(a) = \scal s\prodhat B f(a)$.
  \end{itemize}
\end{definition}

\begin{definition}
  \label{def:arrowsalg}
  Let $A$ be an object in $\TheCatAlg$, we define the following arrows.
  \begin{itemize}
    \item $\{\star\}\xlra{\scal s}\mathcal S$ is defined as $\star\mapsto\scal s$.
    \item $\mathcal S\times A\xlra{\prodhat A}A$ is defined as $(\scal s,a)\mapsto\scal s\prodhat A a$.
    \item $A\xlra{\hat{\scal s}}A$ is defined as $A\xlra{\rho}\{\star\}\times A\xlra{\scal s\times\Id}\mathcal S\times A\xlra{\prodhat A} A$.
  \end{itemize}
\end{definition}

\subsection{Model of the \texorpdfstring{$\alglambda$ calculus}{algebraic lambda calculus}}\label{sec:modelAlg}
\subsubsection{Interpretation}

The interpretation of propositions and contexts from \cref{sec:interpretation} remains the same, except that the interpretation of $\top$ is updated as $\sem\top = \mathcal S$.

The interpretation of the deduction rules is also the same as that from \cref{sec:interpretation}, except for the introduction and elimination of $\top$, and the extra rules prod($\scal s$), which are the following: 
\[
  \begin{array}{r@{\,}l}
  \sem{\Gamma\vdash \sstar s:\top} &= \sem{\Gamma}\xlra{!}\{\star\}\xlra{\scal s}\mathcal S\\
  \sem{\Gamma\vdash\elimtop(t,u):A} &=\sem{\Gamma}\xlra{\Delta}\sem{\Gamma}\times\sem{\Gamma}\xlra{t\times u} \mathcal S\times\sem A\xlra{\prodhat A}\sem A\\
  \sem{\Gamma\vdash \scal s\bullet t:A} &=\sem{\Gamma}\xlra{t}\sem A\xlra{\hat{\scal s}}\sem A
\end{array}
\]

\subsubsection{Soundness}
The proof of soundness is updated from the proof of soundness from \cref{sec:soundness}.

As before, we require the substitution lemma (proved in \appref{D.1}).

\begin{lemma}
  [Substitution]\label{lem:substitutionAlg}
  If ${x:A},\Gamma\vdash {t:B}$ and $\Gamma\vdash {u:A}$, then
  $\sem{\Gamma\vdash (u/x)t:B} = \sem{\Gamma\vdash t:B}\circ(\sem{\Gamma\vdash u:A}\times\Id)\circ\Delta$.
  \qed
\end{lemma}

\begin{theorem}
  [Soundness]\label{thm:soundnessAlg}
  If $t\lra r$ and $\Gamma\vdash t:A$, then
  $\sem{\Gamma\vdash t:A} = \sem{\Gamma\vdash r:A}$.
\end{theorem}
\begin{proof}
  By induction on the relation $\longrightarrow$.
  The full proof can be found in \appref{D.2}.
  We only give a few illustrative examples.

  \noindent  Rule $\elimtop(\sstar s,t) \longrightarrow  \scal s\bullet t$. The commuting diagram is the following:
      \[
	\begin{tikzcd}[row sep=1.3cm]
	  \Gamma & {\Gamma\times \Gamma} & {\{\star\}\times A}  
	  \\
	  A & A & {\mathcal S\times A}
	  \arrow["\Delta", from=1-1, to=1-2]
	  \arrow["t"', from=1-1, to=2-1]
	  \arrow["{{!}\times t}", from=1-2, to=1-3]
	  \arrow["{\scal s\times \Id}", from=1-3, to=2-3]
	  \arrow["\rho", dashed, from=2-1, to=1-3,sloped]
	  \arrow["{\hat{\scal s}}"', from=2-1, to=2-2]
	  \arrow["{\prodhat A}", from=2-3, to=2-2]
	\end{tikzcd}
      \]
	    The commutation of the left diagram is trivial. 
	    The right diagram is the definition of $\hat{\scal s}$.
    Rule $\sstare{\scal s_1}\plus \sstare{\scal s_2} \longrightarrow  \sstare{\scal s_1+\scal s_2}$. The commuting diagram is the following:
      \[
	\begin{tikzcd}[row sep=1.3cm]
	  \Gamma & {\{\star\}} & {\mathcal S} 
	  \\
	  {\Gamma\times\Gamma} & {\{\star\}\times \{\star\}} & {\mathcal S\times \mathcal S}
	  \arrow["{{!}}", from=1-1, to=1-2]
	  \arrow["\Delta"', from=1-1, to=2-1]
	  \arrow["{\scal s_1+\scal s_2}", from=1-2, to=1-3]
	  \arrow["\Delta"', from=1-2, to=2-2,dashed]
	  \arrow["{{!}\times{!}}"', from=2-1, to=2-2]
	  \arrow["{\scal s_1\times\scal s_2}"', from=2-2, to=2-3]
	  \arrow["\sumhat{\mathcal S}"', from=2-3, to=1-3]
	\end{tikzcd}
      \]
	    The left diagram commutes by naturality of $\Delta$.
	    The right diagram commutes since the magma operation of $\mathcal S$ is defined with respect to the bi-magma operation: $\sumhat{\mathcal S}=+$.
    Rule $\scal s\bullet\inl(t) \longrightarrow \inl(\scal s\bullet t)$. The commuting diagram is
      \[
	\begin{tikzcd}[row sep=1.3cm]
	  \Gamma & A & {A\cp B} 
	  \\
	  & A & {A\cp B}
	  \arrow["t", from=1-1, to=1-2]
	  \arrow["{i_1}", from=1-2, to=1-3]
	  \arrow["{\hat{\scal s}}"', from=1-2, to=2-2]
	  \arrow["{\hat{\scal s}}", from=1-3, to=2-3]
	  \arrow["{i_1}"', from=2-2, to=2-3]
	\end{tikzcd}
      \]
	    The diagram commutes since $\scal s\prodcp(a,0) = (\scal s\prodhat A a,0)$ by definition.
\end{proof}

\subsubsection{Adequacy}
In this setting, the concept of computational equivalence is simpler than in the calculus $\lambda_\parallel$, since the proposition $\top\vee\top$ no longer represents the booleans but rather the action magma $\mathcal S^2$. Thus, to distinguish two programs, it is sufficient to use the proposition $\top$, as it contains as many scalars as $\mathcal S$ (unless $\mathcal S=\emptyset$, in which case the language becomes trivial, or $\mathcal S=\{\star\}$, in which case the language is analogous to $\lambda_\parallel$).

\begin{theorem}
  [Adequacy]\label{thm:adequacyAlg}
  If $\sem{\vdash t:A} = \sem{\vdash u:A}$, then $t\sim u$.
\end{theorem}
\begin{proof}
  By induction on the structure of $A$.
  See \appref{D.3}.
\end{proof}

\section{Conclusion}\label{sec:conclusion}
\subparagraph*{Summary}
In this paper, we have proposed a categorical framework for modelling parallelism and algebraic combinations within intuitionistic propositional logic, without relying on monads or biproducts. Our approach is based on the minimal structure of magmas and action magmas, with arrows taken from the category $\mathbf{Set}$, which allows us to interpret two lambda calculi: a parallel lambda calculus $\lambda_\parallel$ and its algebraic extension $\alglambda$.

A central challenge addressed in this work is the interpretation of disjunction in the presence of parallel and scalar operators. 
Unlike in certain models of linear logic---those equipped with biproducts---where it is possible to identify disjunction and conjunction via a common categorical structure, intuitionistic
logic does not provide this feature. Our solution, based on a set-theoretic construction combining the disjoint union and the Cartesian product, allows us to model disjunction while preserving the intended behaviour of parallel composition. This highlights the disjunction as the most delicate connective to handle in this context, and its resolution was crucial to enabling a uniform interpretation of both calculi.

By showing that these effects can be captured without requiring monads or biproducts, we provide an alternative and structurally simpler foundation for reasoning about parallelism in logic. In particular, our work suggests that many of the insights from linear logic---including the idea of interpreting computational effects through categorical structure---can be translated to a broader logical setting, provided that we carefully reconsider the structural assumptions.

\subparagraph*{Homomorphisms}
Neither \TheCat nor \TheCatAlg are algebraic categories, in the sense that the natural arrows for the given objects would be homomorphisms preserving the operations. The reason for this differs for each calculus and deserves some discussion.

In the case of the $\lambda_\parallel$ calculus, we could have opted for a category whose arrows preserve the operations of the objects. However, in such a case, it would not have been sufficient to work with magmas, since the arrow $\sumcp:(A\cp B)\times (A\cp B)\lra A\cp B$ is not a homomorphism unless it is associative and commutative, which is not the case unless $\sumhat A$ and $\sumhat B$ are. Indeed, if $f=\sumcp$, then 
$f((c_1,c_2)\sumhat\times (c_3,c_4)) =f(c_1\sumcp c_3,c_2\sumcp c_4) =(c_1\sumcp c_3)\sumcp(c_2\sumcp c_4)
\neq
(c_1\sumcp c_2)\sumcp(c_3\sumcp c_4)
=
f(c_1,c_2)\sumcp f(c_3,c_4)$.
Moreover, the category of commutative magmas is not Cartesian closed, and thus, we would not be able to interpret the $\lambda_\parallel$ calculus in it.

The advantage of using the category \TheCat is that it allows us to interpret the parallel operator in a more general setting, where the operator does not need to be associative or commutative. It also gives us a closer connection to the algebraic lambda calculus, where the interpretation cannot be done in a category with homomorphisms.
Indeed, it is well known~\cite{AssafDiazcaroPerdrixTassonValironLMCS14} that $(\lambda x.t)(\scal a\bullet u)$ and $\scal a\bullet (\lambda x.t)u$ do not yield the same result, unless the linearity is enforced by the reduction strategy as is the case of Lineal~\cite{ArrighiDowekLMCS17}\footnote{In Lineal, $(\lambda x.t)(\scal a\bullet u)$ does not beta-reduce, but first reduces to $\scal a\bullet (\lambda x.t)u$ until $u$ does not include any more scalars or sums, so, forcing the linearity by evaluation.}. 
An easy counterexample, taking $\mathcal S=\mathbb N$, is the following:
\(
  (\lambda x.\sstare 1)(2\bullet t)
  \longrightarrow 
  \sstare 1 
  \neq
  \sstare 2
  \lla
  2\bullet \sstare 1
  \lla
  2\bullet (\lambda x.\sstare 1)t
\).


\subparagraph*{The interpretation of disjunctions}
As mentioned in the introduction, the fact that the parallel of disjunctions can be seen as a pair is evident. For example, compare $\inl(t)\parallel\inr(u)$ with $\pair{t}{u}$, where $t$ and $u$ are terms of type $A$ and $B$, respectively.
This resemblance becomes even more apparent when considering the Linear Logic additive disjunction and conjunction, which are both interpreted as a biproduct.

The interpretation of disjunction that we gave for the non-linear case, namely $A\cp B = (A \cup B) \cup (A \times B)$, attempts to mimic this biproduct-like behaviour, but without making disjunction and conjunction coincide. Instead, it establishes an inclusion.
Indeed, when linearity is not considered, we cannot simply add a zero to a conjunction to obtain a disjunction.
In a conjunction, both $t$ and $u$ are present, whereas in a disjunction with parallel, both $t$ and $u$ can be present — but it is still possible to have only one of them. This provides a more refined interpretation of disjunction.

We stress that our system does feature a proper disjunction: the only novelty is that,
due to the presence of the parallel operator, terms such as $\inl(v)\parallel\inr(w)$ may
also be viewed as conjunctive. This viewpoint is consistent with the recent
preprint~\cite{DiazcaroDowekInlr}, where instead of a parallel operator an additional
introduction rule $\mathsf{inlr}(v,w)$ is considered, which is equivalent to
$\inl(v)\parallel\inr(w)$.

This perspective also highlights the added value of our model: the introduction of the
pokeball product provides a semantic account of this phenomenon, which had previously
been treated only at the syntactic level. In particular, it captures the situation where both
sides of a disjunction are present simultaneously. The pokeball product shows that this
behaviour can be modelled categorically in intuitionistic logic, without requiring
biproducts, thereby revealing a subtle logical feature of disjunction in the presence of
parallelism.

\subparagraph*{Future directions}
The calculi we presented are strongly normalising. As future work, we plan to investigate
the addition of fixpoint operators to extend their expressive power while ensuring that
desirable semantic properties are preserved.

More broadly, we believe that the approach developed here could be extended to capture
other computational effects beyond parallelism and algebraic combinations, in a spirit
similar to Moggi’s monadic framework. Identifying the appropriate categorical structures
for such effects remains an open and promising direction for future work.

\bibliography{biblio}

\ifisarxiv
  \appendix
  \section[Full proofs from Section~\ref{sec:correctness}]{Full proofs from \cref{sec:correctness}}\label{A}
  This section is a nearly literal reproduction of the proofs from the draft in~\cite{DiazcaroDowekInlr}, with only minimal changes required to adapt them to the $\lambda^\parallel$ calculus.

  \subsection[Proof of Theorem~\ref{thm:SR}]{Proof of \cref{thm:SR}}
  As usual, we start with a substitution lemma
  (\cref{prop:subst}) which is used in
  the proof of the main theorem
  (\cref{thm:SR}).

  The substitution lemma says that substituting a hypothesis in proof, by an
  actual proof of that hypothesis, yields a valid proof.
  \begin{proposition}
    [Substitution]
    \label{prop:subst}
    Let $t$ be a proof of $A$ in context $\Gamma,x:B$ and $u$ be a proof of $B$
    in context $\Gamma$.  Then, $(u/x)t$ is a proof of $A$ in context $\Gamma$.
  \end{proposition}
  \begin{proof}
    By induction on the structure of $t$.
    \begin{itemize}
      \item Let $t=x$, thus $B=A$ and $(u/x)t = u$, and since $u$ is a proof of
	$B$ in context $\Gamma$, we have that $(u/x)t$ is a proof of $A$ in
	context $\Gamma$.
      \item Let $t=y$, thus $y:A\in\Gamma$ and $(u/x)t=y$, hence $(u/x)t$ is a
	proof of $A$ in context $\Gamma$.
      \item Let $t=v \plus w$, 
	thus $(u/x)t=(u/x)v\plus (u/x)w$ and
	$v$ and $w$ are proofs of $A$ in context $\Gamma,x:B$. Hence, by the
	induction hypothesis,
	$(u/x)v$ is a proof of $A$ in context $\Gamma$ and
	$(u/x)w$ is a proof of $A$ in context $\Gamma$. Therefore, 
	$(u/x)t$ is a proof of $A$ in context $\Gamma$.
      \item Let $t=\star$, thus $A=\top$ and $(u/x)t=\star$, hence $(u/x)t$ is a
	proof of $A$ in context $\Gamma$.
      \item Let $t= \elimtop(v,w)$,  
	thus $(u/x)t=\elimtop((u/x)v,(u/x)w)$ and
	$v$ is a proof of $\top$ in the context $\Gamma,x:B$ and $w$ is a proof of $A$ in context $\Gamma,x:B$. Hence, by the induction
	hypothesis,
	$(u/x)v$ is a proof of $\top$ in context $\Gamma$ and
	$(u/x)w$ is a proof of $A$ in context $\Gamma$. Therefore, 
	$(u/x)t$ is a proof of $A$ in context $\Gamma$.
      \item Let $t= \elimbot{A}(v)$, 
	thus $(u/x)t=\elimbot{A}((u/x)v)$ and
	$v$ is a proof of $\bot$ in context $\Gamma,x:B$.
	Hence, by the induction hypothesis,
	$(u/x)v$ is a proof of $\bot$ in context $\Gamma$. Therefore,
	$(u/x)t$ is a proof of $A$ in context $\Gamma$.
      \item Let $t= \lambda \abstr{y}v$,
	thus $(u/x)t=\lambda\abstr{y}((u/x)v)$, $A=C\Rightarrow D$, and
	$v$ is a proof of $D$ in context $\Gamma,x:B,y:C$.
	Hence, by the induction hypothesis,
	$(u/x)v$ is a proof of $D$ in context $\Gamma,y:C$. Therefore,
	$(u/x)t$ is a proof of $A$ in context $\Gamma$.
      \item Let $t=v~w$,
	thus $(u/x)t=(u/x)~v(u/x)w$ and
	$v$ is a proof of $C\Rightarrow A$ in the context $\Gamma,x:B$ and $w$ is a proof of $C$ in context $\Gamma,x:B$. Hence, by the induction
	hypothesis,
	$(u/x)v$ is a proof of $C\Rightarrow A$ in context $\Gamma$ and
	$(u/x)w$ is a proof of $C$ in context $\Gamma$. Therefore, 
	$(u/x)t$ is a proof of $A$ in context $\Gamma$.
      \item Let $t= \pair{v}{w}$,
	thus $(u/x)t=\pair{(u/x)v}{(u/x)w}$, $A=C\wedge D$, 
	$v$ is a proof of $C$ in the context $\Gamma,x:B$, 
	and $w$ is a proof of $D$ in context $\Gamma,x:B$. Hence, by the
	induction hypothesis,
	$(u/x)v$ is a proof of $C$ in context $\Gamma$ and
	$(u/x)w$ is a proof of $D$ in context $\Gamma$. Therefore, 
	$(u/x)t$ is a proof of $A$ in context $\Gamma$.
      \item Let $t= \pi_1(v)$,
	thus $(u/x)t=\pi_1((u/x)v)$ and
	$v$ is a proof of $C\wedge D$ in the context $\Gamma,x:B$ 
	Hence, by the
	induction hypothesis,
	$(u/x)v$ is a proof of $C\wedge D$ in context $\Gamma$.
	Therefore, 
	$(u/x)t$ is a proof of $A$ in context $\Gamma$.
      \item Let $t= \pi_2(v)$. This case is  analogous to the previous case.
      \item Let $t= \inl(v)$,
	thus $(u/x)t=\inl((u/x)v)$, $A=C\vee D$, and
	$v$ is a proof of $C$ in context $\Gamma,x:B$.
	Hence, by the induction hypothesis,
	$(u/x)v$ is a proof of $C$ in context $\Gamma,y:C$. Therefore,
	$(u/x)t$ is a proof of $A$ in context $\Gamma$.
      \item Let $t= \inr(v)$. This case is  analogous to the previous case.
      \item Let $t= \elimor(v,\abstr{y_1}w_1,\abstr{y_2}w_2)$,
	thus $(u/x)t=\elimor({(u/x)v},\abstr{y_1}{(u/x)w_1},\abstr{y_2}w_2)$, 
	$v$ is a proof of $C\vee D$ in the context $\Gamma,x:B$, 
	$w_1$ is a proof of $A$ in context $\Gamma,x:B,y_1:C$, and
	$w_2$ is a proof of $A$ in context $\Gamma,x:B,y_2:D$.
	Hence, by the induction hypothesis,
	$(u/x)v$ is a proof of $C\vee D$ in context $\Gamma$,
	$(u/x)w_1$ is a proof of $A$ in context $\Gamma,y_1;C$, and
	$(u/x)w_2$ is a proof of $A$ in context $\Gamma,y_2;D$.
	Therefore, $(u/x)t$ is a proof of $A$ in context $\Gamma$.
	\qedhere
    \end{itemize}
  \end{proof}

  \begin{proof}[Proof of \cref{thm:SR}]
    By induction on the reduction relation. The inductive cases are trivial, so we
    only treat the basic cases corresponding to reduction rules~\eqref{ruelimtop}
    to~\eqref{rusuminlr}.
    \begin{enumerate}
      \item Let $t=\elimtop(\star, v)$ and $u=v$. Then, by inversion $u$ is a
	proof of $A$ in the context $\Gamma$.
      \item Let $t=(\lambda \abstr{x}v)~w$ and $u=(w/x)v$. Then $v$ is a proof
	of $A$ in the context $\Gamma,x:B$ and $w$ is a proof of $B$ in the
	context $\Gamma$. Thus, by \cref{prop:subst}, we have that
	$(w/x)v$ is a proof of $A$ in the context $\Gamma$.
      \item Let $t=\pi_1\pair{v_1}{v_2}$ and $u=v_1$.
	Then $v_1$ is a proof of $A$ in the context $\Gamma$.
      \item Let $t=\pi_2\pair{v_1}{v_2}$.
	This case is analogous to the previous case.
      \item Let $t=\elimor(\inl(v),\abstr{x}w_1,\abstr{y}w_2)$ and $u=(v/x)w_1$.
	Then $v$ is a proof of $B$ in the context $\Gamma$ and $w_1$ is a proof
	of $A$ in the context $\Gamma,x:B$. Thus, by
	\cref{prop:subst}, we have that $(v/x)w_1$ is a proof of $A$
	in the context $\Gamma$.
      \item Let $t=\elimor(\inr(v),\abstr{x}w_1,\abstr{y}w_2)$ and $u=(v/y)w_2$.
	This case is analogous to the previous case.
      \item Let $t=\elimor(\inl(v_1)\plus\inr(v_2),\abstr{x}w_1,\abstr{y}w_2)$ and $u=(v_1/x)w_1 \plus (v_2/y)w_2$.
	Then 
	$v_1$ is a proof of $B$ in the context $\Gamma$,
	$v_2$ is a proof of $C$ in the context $\Gamma$,
	$w_1$ is a proof of $A$ in the context $\Gamma,x:B$, and
	$w_2$ is a proof of $A$ in the context $\Gamma,y:C$.
	Thus, by \cref{prop:subst}, we have that
	$(v_1/x)w_1$ is a proof of $A$ in the context $\Gamma$ and 
	$(v_2/x)w_2$ is a proof of $A$ in the context $\Gamma$.
	Thus, $(v_1/x)w_1 \plus (v_2/y)w_2$ is a proof of $A$ in the context $\Gamma$.
      \item Let $t=\star \plus \star$ and $u=\star$. Then $A=\top$, and we have
	that $\star$ is a proof of $\top$ in the context $\Gamma$.
      \item Let $t=(\lambda \abstr{x}v) \plus (\lambda \abstr{x}w)$ and $u=\lambda \abstr{x}(v \plus w)$.
	Then, $A=B\Rightarrow C$,
	$v$ is a proof of $C$ in the context $\Gamma,x:B$, and
	$w$ is a proof of $C$ in the context $\Gamma,x:B$.
	Therefore, $v\plus w$ is a proof of $C$ in the context $\Gamma,x:B$,
	and so $\lambda \abstr{x}(v \plus w)$ is a proof of $A$ in the context
	$\Gamma$.
      \item Let $t=\pair{v_1}{w_1} \plus \pair{v_2}{w_2}$ and $u=\pair{v_1\plus v_2}{w_1\plus w_2}$.
	Then, $A=B\wedge C$,
	$v_1$ is a proof of $B$ in the context $\Gamma$,
	$v_2$ is a proof of $B$ in the context $\Gamma$,
	$w_1$ is a proof of $C$ in the context $\Gamma$, and
	$w_2$ is a proof of $C$ in the context $\Gamma$.
	Thus, $\pair{v_1\plus v_2}{w_1\plus w_2}$ is a proof of $A$ in the context
	$\Gamma$.
      \item Let $t=\inl(v_1) \plus \inl(v_2)$ and $u=\inl(v_1\plus v_2)$
	Then, $A=B\vee C$,
	$v_1$ is a proof of $B$ in the context $\Gamma$,
	$v_2$ is a proof of $B$ in the context $\Gamma$,
	Thus, $\inl{v_1\plus v_2}$ is a proof of $A$ in the context $\Gamma$.
      \item Let $t=\inl(v_1) \plus \inl(v_2)\plus\inr(w)$ and $u=\inl(v_1 \plus v_2)\plus\inr(w) $
	Then, $A=B\vee C$,
	$v_1$ is a proof of $B$ in the context $\Gamma$,
	$v_2$ is a proof of $B$ in the context $\Gamma$,
	$w$ is a proof of $C$ in the context $\Gamma$, and
	Thus, $\inl(v_1\plus v_2)\plus\inr(w)$ is a proof of $A$ in the context $\Gamma$.
      \item Let $t=\inr(w) \plus \inl(v)$ and $u=\inl(v)\plus\inr(w)$
	Then, $A=B\vee C$,
	$v$ is a proof of $B$ in the context $\Gamma$,
	$w$ is a proof of $C$ in the context $\Gamma$, and
	Thus, $\inl(v)\plus\inr(w)$ is a proof of $A$ in the context $\Gamma$.
      \item Let $t=\inr(w_1) \plus \inr(w_2)$ and $u=\inr(w_1 \plus w_2) $
	Then, $A=B\vee C$,
	$w_1$ is a proof of $C$ in the context $\Gamma$, and
	$w_2$ is a proof of $C$ in the context $\Gamma$.
	Thus, $\inr{w_1\plus w_2}$ is a proof of $A$ in the context $\Gamma$.
      \item Let $t=\inr(w_1) \plus \inl(v)\plus\inr(w_2)$ and $u=\inl(v)\plus\inr(w_1 \plus w_2) $
	Then, $A=B\vee C$,
	$v$ is a proof of $B$ in the context $\Gamma$,
	$w_1$ is a proof of $C$ in the context $\Gamma$, and
	$w_2$ is a proof of $C$ in the context $\Gamma$.
	Thus, $\inl(v)\plus\inr(w_1\plus w_2)$ is a proof of $A$ in the context $\Gamma$.
      \item Let $t=\inl(v_1)\plus\inr(w) \plus \inl(v_2) $ and $u= \inl(v_1 \plus v_2)\plus\inr(u)$
	Then, $A=B\vee C$,
	$v_1$ is a proof of $B$ in the context $\Gamma$,
	$v_2$ is a proof of $B$ in the context $\Gamma$,
	$w$ is a proof of $C$ in the context $\Gamma$, and
	Thus, $\inl(v_1\plus v_2)\plus\inr(w)$ is a proof of $A$ in the context $\Gamma$.
      \item Let $t=\inl(v)\plus\inr(w_1) \plus \inr(w_2) $ and $u= \inl(v)\plus\inr(w_1 \plus w_2) $
	Then, $A=B\vee C$,
	$v$ is a proof of $B$ in the context $\Gamma$,
	$w_1$ is a proof of $C$ in the context $\Gamma$, and
	$w_2$ is a proof of $C$ in the context $\Gamma$.
	Thus, $\inl(v)\plus\inr(w_1\plus w_2)$ is a proof of $A$ in the context $\Gamma$.
      \item Let $t=\inl(v_1)\plus\inr(w_1) \plus \inl(v_2)\plus\inr(w_2) $ and $u= \inl(v_1 \plus v_2)\plus\inr(w_1 \plus w_2)$
	Then, $A=B\vee C$,
	$v_1$ is a proof of $B$ in the context $\Gamma$,
	$v_2$ is a proof of $B$ in the context $\Gamma$,
	$w_1$ is a proof of $C$ in the context $\Gamma$, and
	$w_2$ is a proof of $C$ in the context $\Gamma$.
	Thus, $\inl(v_1\plus v_2)\plus\inr(w_1\plus w_2)$ is a proof of $A$ in the context $\Gamma$.
	\qedhere
    \end{enumerate}
  \end{proof}

  \subsection[Proof of Theorem~\ref{thm:SN}]{Proof of \cref{thm:SN}}

  \begin{definition}[Length of reduction]
    If $t$ is a strongly terminating proof, we write $|t|$ for the
    maximum length of a reduction sequence issued from $t$.
  \end{definition}

  \begin{proposition}[Termination of a parallel]
    \label{terminationsum}
    If $t$ and $u$ strongly terminate, then so does $t \plus u$. 
  \end{proposition}

  \begin{proof}
    We prove that all the one-step reducts of $t \plus u$ strongly
    terminate, by induction first on $|t| + |u|$ and then on the size of
    $t$.

    If the reduction takes place in $t$ or in $u$ we apply the induction
    hypothesis.
    Otherwise, the reduction is at the root and the rule used is one of
    the rules
    \eqref{rusumstar} to \eqref{rusuminlr}.

    In the case \eqref{rusumstar}, the proof $\star$ is irreducible,
    hence it strongly terminates.

    In the case \eqref{rusumlam}, $t = \lambda \abstr{x} t_1$, $u =
    \lambda \abstr{x} u_1$, by induction hypothesis, the proof $t_1 \plus
    u_1$ strongly terminates, thus so does the proof $\lambda
    \abstr{x}(t_1 \plus u_1)$.  The proof is similar for the cases
    \eqref{rusumpair} to \eqref{rusuminlr}.  \qedhere
  \end{proof}

  \begin{definition}
    \label{def:interpretation}
    We define, by induction on the proposition $A$, a set of proofs
    $\llbracket A \rrbracket$:
    \begin{itemize}
      \item $t \in \llbracket \top \rrbracket$ if $t$ strongly terminates,

      \item $t \in \llbracket \bot \rrbracket$ if $t$ strongly terminates,

      \item $t \in \llbracket A \Rightarrow B \rrbracket$ if $t$ strongly
	terminates and whenever it reduces to a proof of the form $\lambda
	\abstr{x}u$, then for every $v \in \llbracket A \rrbracket$, $(v/x)u \in
	\llbracket B \rrbracket$,

      \item $t \in \llbracket A \wedge B \rrbracket$ if $t$ strongly
	terminates and whenever it reduces to a proof of the form
	$\pair{u}{v}$, then $u \in \llbracket A \rrbracket$ and $v \in \llbracket
	B \rrbracket$,

      \item $t \in \llbracket A \vee B \rrbracket$ if $t$ strongly
	terminates and whenever it reduces to a proof of the form
	$\inl(u)$ (resp. $\inr(v)$, $\inl(u)\plus\inr(v)$), then $u \in
	\llbracket A \rrbracket$ (resp. $v \in \llbracket B \rrbracket$,
	$u \in \llbracket A \rrbracket$ and $v \in \llbracket B \rrbracket$).
    \end{itemize}
  \end{definition}

  \begin{proposition}[Variables]
    \label{Var}
    For any $A$, the set $\llbracket A \rrbracket$ contains all the variables.
  \end{proposition}

  \begin{proof}
    A variable is irreducible, hence it strongly terminates. Moreover, it
    never reduces to an introduction.
    \qedhere
  \end{proof}   

  \begin{proposition}[Closure by reduction]
    \label{closure}
    If $t \in \llbracket A \rrbracket$ and $t \longrightarrow^* t'$, then 
    $t' \in \llbracket A \rrbracket$.
  \end{proposition}

  \begin{proof}
    If $t \longrightarrow^* t'$ and $t$ strongly terminates, then $t'$
    strongly terminates.

    Furthermore, if $A$ has the form $B \Rightarrow C$ and $t'$ reduces to
    $\lambda \abstr{x}u$, then so does $t$, hence for every $v \in \llbracket B
    \rrbracket$, $(v/x)u \in \llbracket C \rrbracket$.

    If $A$ has the form $B \wedge C$ and $t'$ reduces to $\pair{u}{v}$,
    then so does $t$, hence $u \in \llbracket B \rrbracket$ and $v \in
    \llbracket C \rrbracket$.

    If $A$ has the form $B \vee C$ and $t'$ reduces to $\inl(u)$
    (resp. $\inr(v)$, $\inl(u)\plus\inr(v)$), then so does $t$, hence
    $u \in   \llbracket A \rrbracket$ (resp. $v \in \llbracket B \rrbracket$,
    $u \in   \llbracket A \rrbracket$ and $v \in \llbracket B \rrbracket$).
    \qedhere
  \end{proof}

  \begin{proposition}[Girard's lemma]
    \label{CR3}
    Let $t$ be a proof that is not an introduction, 
    such that all the one-step reducts of $t$
    are in $\llbracket A \rrbracket$. Then, $t \in \llbracket A \rrbracket$.
  \end{proposition}

  \begin{proof}
    Let $t, t_2, \ldots$ be a reduction sequence issued from $t$. If it has a
    single element, it is finite. Otherwise, we have $t \longrightarrow
    t_2$. As $t_2 \in \llbracket A \rrbracket$, it strongly terminates and
    the reduction sequence is finite. Thus, $t$ strongly terminates.

    Furthermore, if $A$ has the form $B \Rightarrow C$ and $t
    \longrightarrow^* \lambda \abstr{x}u$, then let $t , t_2, \ldots, t_n$ be a
    reduction sequence from $t$ to $\lambda \abstr{x}u$.  As $t_n$ is an
    introduction and $t$ is not, $n \geq 2$. Thus, $t \longrightarrow t_2
    \longrightarrow^* t_n$. We have $t_2 \in \llbracket A \rrbracket$,
    thus for all $v \in \llbracket B \rrbracket$, $(v/x)u \in \llbracket C
    \rrbracket$.

    If $A$ has the form $B \wedge C$ and $t \longrightarrow^* \pair{u}{v}$, then let $t , t_2, \ldots, t_n$ be a reduction sequence
    from $t$ to $\pair{u}{v}$.  As $t_n$ is an introduction and
    $t$ is not, $n \geq 2$. Thus, $t \longrightarrow t_2 \longrightarrow^*
    t_n$. We have $t_2 \in \llbracket A \rrbracket$, thus $u \in
    \llbracket B \rrbracket$ and $v \in \llbracket C \rrbracket$.

    And if $A$ has the form $B \vee C$ and $t \longrightarrow^* 
    \inl(u)$ (resp. $\inr(v)$, $\inl(u)\plus\inr(v)$)
    then
    let $t , t_2, \ldots, t_n$ be a reduction sequence from $t$ to 
    $\inl(u)$ (resp. $\inr(v)$,  $\inl(u)\plus\inr(v)$).
    As $t_n$ is an introduction and $t$ is not, $n \geq 2$. Thus, $t
    \longrightarrow t_2 \longrightarrow^* t_n$. We have $t_2 \in
    \llbracket A \rrbracket$, thus 
    $u \in \llbracket B \rrbracket$ (resp. $v \in \llbracket C \rrbracket$, 
    $u \in \llbracket B \rrbracket$ and $v \in \llbracket C \rrbracket$).
    \qedhere
  \end{proof}

  In \cref{sum,star,abstraction,pair,inl,inr,elimtop,elimbot,application,elimand1,elimand2,elimor}, we prove the adequacy of each proof constructor.

  \begin{proposition}[Adequacy of $\plus$]
    \label{sum}
    If $t_1 \in \llbracket A \rrbracket$ and $t_2 \in \llbracket A
    \rrbracket$, then $t_1 \plus t_2 \in \llbracket A \rrbracket$.
  \end{proposition}

  \begin{proof}

    By induction on $A$.

    The proofs $t_1$ and $t_2$ strongly terminate.  Thus, by
    \cref{terminationsum}, the proof $t_1 \plus t_2$ strongly
    terminates.

    Furthermore:
    \begin{itemize}
      \item If the proposition $A$ has the form $\top$ or $\bot$, then
	the strong termination of $t_1 \plus t_2$ is sufficient to conclude
	that $t_1 \plus t_2 \in \llbracket A \rrbracket$.

      \item If the proposition $A$ has the form $B \Rightarrow C$, and $t_1
	\plus t_2 \lra^* \lambda \abstr{x} v$ then $t_1 \lra^* \lambda
	\abstr{x} u_1$, $t_2 \lra^* \lambda \abstr{x} u_2$, and $u_1 \plus
	u_2 \lra^* v$.

	As $t_1 \in \llbracket B \Rightarrow C \rrbracket$ and $t_1 \lra^*
	\lambda \abstr{x} u_1$, we have, for every $w$ in $\llbracket B
	\rrbracket$, $(w/x)u_1 \in \llbracket C \rrbracket$.  In the same
	way, for every $w$ in $\llbracket B \rrbracket$, $(w/x)u_2 \in
	\llbracket C \rrbracket$.  By induction hypothesis, $(w/x)(u_1 \plus
	u_2) = (w/x)u_1 \plus (w/x)u_2 \in \llbracket C \rrbracket$ and by
	\cref{closure}, $(w/x)v \in \llbracket C \rrbracket$.

      \item If the proposition $A$ has the form $B \wedge C$, and $t_1
	\plus t_2 \lra^* \pair{v}{v'}$ then
	$t_1 \lra^* \pair{u_1}{u'_1}$,
	$t_2 \lra^* \pair{u_2}{u'_2}$, $u_1 \plus u_2 \lra^* v$, and $u'_1
	\plus u'_2 \lra^* v'$.

	As $t_1 \in \llbracket B \wedge C \rrbracket$ and $t_1 \lra^*
	\pair{u_1}{u'_1}$, we have $u_1 \in \llbracket B \rrbracket$
	and $u'_1 \in \llbracket C \rrbracket$. In the same way, 
	we have $u_2 \in \llbracket B \rrbracket$
	and $u'_2 \in \llbracket C \rrbracket$. 
	By induction hypothesis, $u_1 \plus u_2 \in \llbracket B \rrbracket$
	and $u'_1 \plus u'_2 \in \llbracket C \rrbracket$, and, by
	\cref{closure}, $v \in \llbracket B \rrbracket$ and $v' \in
	\llbracket C \rrbracket$.

      \item If the proposition $A$ has the form $B \vee C$, and $t_1 \plus
	t_2 \lra^* \inl(v)$ then $t_1 \lra^* \inl(u_1)$, $t_2 \lra^*
	\inl(u_2)$, and $u_1 \plus u_2 \lra^* v$.  As $t_1 \in \llbracket B
	\vee C \rrbracket$ and $t_1 \lra^* \inl(u_1)$, we have $u_1 \in
	\llbracket B \rrbracket$. In the same way, $u_2 \in \llbracket B
	\rrbracket$.  By induction hypothesis, $u_1 \plus u_2 \in \llbracket
	B \rrbracket$ and, by \cref{closure}, $v \in \llbracket B
	\rrbracket$.

      \item If the proposition $A$ has the form $B \vee C$, and $t_1 \plus
	t_2 \lra^* \inr(v)$ the proof is similar.

      \item
	If the proposition $A$ has the form $B \vee C$, and $t_1 \plus
	t_2 \lra^* \inl(v)\plus\inr(v')$ then, either
	\begin{itemize}
	  \item
	    $t_1 \lra^* \inl(u_1)$, $t_2 \lra^* \inr(u'_2)$, and
	    $u_1 \lra^* v$ and  $u'_2 \lra^* v'$;

	  \item
	    $t_1 \lra^* \inl(u_1)$, $t_2 \lra^* \inl(u_2)\plus\inr(u'_2)$, and
	    $u_1 \plus u_2 \lra^* v$ and  $u'_2 \lra^* v'$;

	  \item
	    $t_1 \lra^* \inr(u'_1)$, $t_2 \lra^* \inl(u_2)$, and
	    $u_2 \lra^* v$ and  $u'_1 \lra^* v'$;

	  \item
	    $t_1 \lra^* \inr(u'_1)$, $t_2 \lra^* \inl(u_2)\plus\inr(u'_2)$, and
	    $u_2 \lra^* v$ and  $u'_1 \plus u'_2 \lra^* v'$;

	  \item
	    $t_1 \lra^* \inl(u_1)\plus\inr(u'_1)$, $t_2 \lra^* \inl(u_2)$, and
	    $u_1 \plus u_2 \lra^* v$ and  $u'_1 \lra^* v'$;

	  \item
	    $t_1 \lra^* \inl(u_1)\plus\inr(u'_1)$, $t_2 \lra^* \inr(u'_2)$, and
	    $u_1 \lra^* v$ and  $u'_1 \plus u'_2 \lra^* v'$;

	  \item
	    or $t_1 \lra^* \inl(u_1)\plus\inr(u'_1)$, $t_2 \lra^* \inl(u_2)\plus\inr(u'_2)$, and
	    $u_1 \plus u_2 \lra^* v$ and  $u'_1 \plus u'_2 \lra^* v'$.
	\end{itemize}

	As these seven cases are similar, we consider only the last one.

	As $t_1 \in \llbracket B \vee C \rrbracket$ and $t_1 \lra^*
	\inl(u_1)\plus\inr(u'_1)$, we have 
	$u_1 \in \llbracket B \rrbracket$
	and
	$u'_1 \in \llbracket C \rrbracket$.
	In the same way, 
	$u_2 \in \llbracket B \rrbracket$
	and
	$u'_2 \in \llbracket C \rrbracket$.
	By induction hypothesis,
	$u_1 \plus u_2 \in \llbracket B \rrbracket$
	and
	$u'_1 \plus u'_2 \in \llbracket C \rrbracket$.
	As $u_1 \plus u_2 \lra^* v$ and $u'_1\plus u'_2\to v'$,
	we get by  \cref{closure}, 
	$v \in \llbracket B \rrbracket$ and
	$v' \in \llbracket C \rrbracket$.
	\qedhere
    \end{itemize}
  \end{proof}

  \begin{proposition}[Adequacy of $\star$]
    \label{star}
    We have $\star \in \llbracket \top \rrbracket$.
  \end{proposition}

  \begin{proof}
    As $\star$ is irreducible, it strongly terminates, hence
    $\star \in \llbracket \top \rrbracket$. 
    \qedhere
  \end{proof}

  \begin{proposition}[Adequacy of $\lambda$]
    \label{abstraction}
    If, for all $u \in \llbracket A \rrbracket$, $(u/x)t \in \llbracket B
    \rrbracket$, then $\lambda \abstr{x}t \in \llbracket A \Rightarrow B
    \rrbracket$.
  \end{proposition}

  \begin{proof}
    By \cref{Var}, $x \in \llbracket A \rrbracket$, thus
    $t = (x/x)t \in \llbracket B \rrbracket$. Hence, $t$ strongly
    terminates.  Consider a reduction sequence issued from $\lambda
    \abstr{x}t$.  This sequence can only reduce $t$ hence it is finite. Thus,
    $\lambda \abstr{x}t$ strongly terminates.

    Furthermore, if $\lambda \abstr{x}t \longrightarrow^* \lambda \abstr{x}t'$, then
    $t \lra^* t'$.  Let $u \in \llbracket A \rrbracket$,
    $(u/x)t \lra^* (u/x)t'$. 
    By \cref{closure}, $(u/x)t' \in \llbracket B \rrbracket$.
    \qedhere
  \end{proof}

  \begin{proposition}[Adequacy of $\pair{}{}$]
    \label{pair}
    If $t_1 \in \llbracket A \rrbracket$ and $t_2 \in \llbracket B
    \rrbracket$, then $\pair{t_1}{t_2} \in \llbracket A \wedge B
    \rrbracket$.
  \end{proposition}

  \begin{proof}
    The proofs $t_1$ and $t_2$ strongly terminate. Consider a reduction
    sequence issued from $\pair{t_1}{t_2}$.  This sequence can only
    reduce $t_1$
    and $t_2$, hence it is finite.  Thus, $\pair{t_1}{t_2}$
    strongly terminates.

    Furthermore, if $\pair{t_1}{t_2} \longrightarrow^* \pair
    {t'_1}{t'_2}$, then $t_1 \lra^* t'_1$ and $t_2 \lra^* t'_2$.  By
    \cref{closure}, $t'_1 \in \llbracket A \rrbracket$ and
    $t'_2 \in \llbracket B \rrbracket$.
    \qedhere
  \end{proof}

  \begin{proposition}[Adequacy of $\inl$]
    \label{inl}
    If $t \in \llbracket A \rrbracket$, then $\inl(t) \in \llbracket A
    \vee B \rrbracket$.
  \end{proposition}

  \begin{proof}
    The proof $t$ strongly terminates. Consider a reduction
    sequence issued from $\inl(t)$.  This sequence can only
    reduce $t$, hence it is finite.  Thus, $\inl(t)$
    strongly terminates.

    Furthermore, if $\inl(t) \longrightarrow^* \inl
    (t')$, then $t \lra^* t'$. By 
    \cref{closure}, $t' \in \llbracket A \rrbracket$.
    And the proof $\inl(t)$ never reduces to a proof of the form $\inr(t'_2)$ or
    $\inl(t'_1)\plus\inr(t'_2)$.
    \qedhere
  \end{proof}

  \begin{proposition}[Adequacy of $\inr$]
    \label{inr}
    If $t \in \llbracket B
    \rrbracket$, then $\inr(t) \in \llbracket A \vee B
    \rrbracket$.
  \end{proposition}

  \begin{proof}
    Similar to that of \cref{inl}.
    \qedhere
  \end{proof}

  \begin{proposition}[Adequacy of $\elimtop$]
    \label{elimtop}
    If $t_1 \in \llbracket \top \rrbracket$ and $t_2 \in \llbracket C \rrbracket$, 
    then $\elimtop(t_1,t_2) \in \llbracket C \rrbracket$.
  \end{proposition}

  \begin{proof}
    The proofs $t_1$ and $t_2$ strongly terminate.  We prove, by
    induction on $|t_1| + |t_2|$, that $\elimtop(t_1,t_2)
    \in \llbracket C \rrbracket$.  Using \cref{CR3}, we only
    need to prove that every of its one step reducts is in $\llbracket C
    \rrbracket$.  If the reduction takes place in $t_1$ or $t_2$, then we
    apply \cref{closure} and the induction hypothesis.

    Otherwise, the proof $t_1$ is $\star$ and the
    reduct is $t_2$. 
    \qedhere
  \end{proof}

  \begin{proposition}[Adequacy of $\elimbot{C}$]
    \label{elimbot}
    If $t \in \llbracket \bot \rrbracket$, 
    then $\elimbot{C}(t) \in \llbracket C \rrbracket$.
  \end{proposition}

  \begin{proof}
    The proof $t$ strongly terminates.  Consider a reduction sequence
    issued from $\elimbot{C}(t)$.  This sequence can only reduce $t$, hence it
    is finite.  Thus, $\elimbot{C}(t)$ strongly terminates.  Moreover, it
    never reduces to an introduction.
    \qedhere
  \end{proof}

  \begin{proposition}[Adequacy of application]
    \label{application}
    If $t_1 \in \llbracket A \Rightarrow B \rrbracket$ and $t_2 \in
    \llbracket A \rrbracket$, then $t_1~t_2 \in \llbracket B
    \rrbracket$.
  \end{proposition}

  \begin{proof}
    The proofs $t_1$ and $t_2$ strongly terminate. We prove, by induction
    on $|t_1| + |t_2|$, that $t_1~t_2 \in \llbracket B \rrbracket$. Using
    \cref{CR3}, we only need to prove that every of its one
    step reducts is in $\llbracket B \rrbracket$.  If the reduction takes
    place in $t_1$ or in $t_2$, then we apply \cref{closure}
    and the induction hypothesis.

    Otherwise, the proof $t_1$ has the form $\lambda \abstr{x}u$ and the reduct
    is $(t_2/x)u$.  As $\lambda \abstr{x}u \in \llbracket A \Rightarrow B
    \rrbracket$, we have $(t_2/x)u \in \llbracket B \rrbracket$.
    \qedhere
  \end{proof}

  \begin{proposition}[Adequacy of $\pi_1$]
    \label{elimand1}
    If $t \in \llbracket A \wedge B \rrbracket$ 
    then $\pi_1(t) \in \llbracket A \rrbracket$.
  \end{proposition}

  \begin{proof}
    The proof $t$ strongly terminates. 
    We prove, by induction on $|t|$, that $\pi_1(t) \in \llbracket A \rrbracket$.  
    Using \cref{CR3}, we only
    need to prove that every of its one step reducts is in $\llbracket A
    \rrbracket$.  If the reduction takes place in $t$, then we
    apply \cref{closure} and the induction hypothesis.

    Otherwise, the proof $t$ has the form $\pair{u}{v}$ and the
    reduct is $u$.  As $\pair{u}{v} \in \llbracket A
    \wedge B \rrbracket$, we have $u \in \llbracket A \rrbracket$.
    \qedhere
  \end{proof}

  \begin{proposition}[Adequacy of $\pi_2$]
    \label{elimand2}
    If $t \in \llbracket A \wedge B \rrbracket$ and,
    then $\pi_2(t) \in \llbracket B \rrbracket$.
  \end{proposition}

  \begin{proof}
    Similar to that of \cref{elimand1}.
    \qedhere
  \end{proof}

  \begin{proposition}[Adequacy of $\elimor$]
    \label{elimor}
    If $t_1 \in \llbracket A \vee B \rrbracket$, for all $u$ in $\llbracket A
    \rrbracket$, $(u/x)t_2 \in \llbracket C \rrbracket $, and, for all $v$
    in $\llbracket B \rrbracket$, $(v/y)t_3 \in \llbracket C \rrbracket $,
    then $\elimor(t_1, \abstr{x}t_2, \abstr{y}t_3) \in \llbracket C \rrbracket$.
  \end{proposition}

  \begin{proof}
    By \cref{Var}, $x \in \llbracket A \rrbracket$, thus $t_2 =
    (x/x)t_2 \in \llbracket C \rrbracket$. In the same way, $t_3 \in
    \llbracket C \rrbracket$.  Hence, $t_1$, $t_2$, and $t_3$ strongly
    terminate.  We prove, by induction on $|t_1| + |t_2| + |t_3|$, that
    $\elimor(t_1, \abstr{x}t_2, \abstr{y}t_3) \in \llbracket C
    \rrbracket$.  Using \cref{CR3}, we only need to prove that
    every of its one step reducts is in $\llbracket C \rrbracket$.  If the
    reduction takes place in $t_1$, $t_2$, or $t_3$, then we apply
    \cref{closure} and the induction hypothesis.

    Otherwise, the proof $t_1$ has the form $\inl(w_2)\plus\inr(w_3)$
    and the reduct is $(w_2/x)t_2 \plus (w_3/x)t_3$,
    or the proof $t_1$ has the form
    $\inl(w_2)$ and the reduct is $(w_2/x)t_2$, or the proof $t_1$ has the form
    $\inr(w_3)$ and the reduct is $(w_3/x)t_3$.

    In the first case, as $\inl(w_2)\plus\inr(w_3) \in \llbracket A \vee B \rrbracket$ we
    have $w_2 \in \llbracket A \rrbracket$ and $w_3 \in \llbracket B \rrbracket$.
    Hence, $(w_2/x)t_2 \in \llbracket C \rrbracket$ and $(w_3/y)t_3 \in
    \llbracket C \rrbracket$.  Moreover, by \cref{sum}, $(w_2/x)t_2
    \plus (w_3/x)t_3 \in \llbracket C \rrbracket$.

    In the second, as $\inl(w_2) \in \llbracket A \vee B \rrbracket$, we
    have $w_2 \in \llbracket A \rrbracket$.  Hence, $(w_2/x)t_2 \in
    \llbracket C \rrbracket$.

    In the third, as $\inr(w_3) \in \llbracket A \vee B \rrbracket$, we
    have $w_3 \in \llbracket B \rrbracket$.  Hence, $(w_3/x)t_3 \in
    \llbracket C \rrbracket$.
    \qedhere
  \end{proof}

  \begin{proposition}[Adequacy]
    \label{prop:adequacy}
    Let $t$ be a proof of $A$ in a context $\Gamma = x_1:A_1, \ldots, x_n:A_n$ and
    $\sigma$ be a substitution mapping each variable $x_i$ to an element
    of $\llbracket A_i \rrbracket$, then $\sigma t \in \llbracket A
    \rrbracket$.
  \end{proposition}

  \begin{proof}
    By induction on the structure of $t$.

    If $t$ is a variable, then, by definition of $\sigma$, $\sigma t \in
    \llbracket A \rrbracket$.  For the other proof constructors, we use
    the \cref{sum,star,abstraction,pair,inl,inr,elimtop,elimbot,application,elimand1,elimand2,elimor}.
    \qedhere
  \end{proof}

  \begin{proof}[Proof of \cref{thm:SN}]
    Let $\sigma$ be the substitution mapping each variable $x_i:A_i$ of
    $\Gamma$ to
    itself. Note that, by \cref{Var}, this variable is an
    element of $\llbracket A_i \rrbracket$.  Then, $t = \sigma t$ is an
    element of $\llbracket A \rrbracket$. Hence, it strongly terminates.
    \qedhere
  \end{proof}

  \subsection[Proof of Theorem~\ref{thm:IP}]{Proof of \cref{thm:IP}}
  \begin{proof}[Proof of \cref{thm:IP}]
    Let $t$ be a closed irreducible proof of some proposition $A$. We prove,
    by induction on the structure of $t$ that $t$ is an introduction.

    As the proof $t$ is closed, it is not a variable.

    It cannot be a sum $u \plus v$, as if it were $u$ and $v$ would be
    closed irreducible proofs of the same proposition, hence, by induction
    hypothesis, they would either be both introductions of $\top$, both
    introductions of $\Rightarrow$, both introductions of $\wedge$, or
    both introductions of $\vee$, and the proof $t$ would be reducible.

    It cannot be an elimination as if it were of the form $\elimtop(u,v)$,
    $u~v$, $\pi_1(u)$, $\pi_2(u)$, or
    $\elimor(u,\abstr{x}v,\abstr{y}w)$, then $u$ would a closed
    irreducible proof, hence, by induction hypothesis, it would an
    introduction and the proof $t$ would be reducible.  If it were
    $\elimbot(u)$, then $u$ would be a closed irreducible proof of $\bot$,
    by induction hypothesis, it would be an introduction of $\bot$ and no
    such introduction rule exists.

    Hence, it is an introduction. \qedhere
  \end{proof}

  \section[Full proofs from Section~\ref{sec:model}]{Full proofs from \cref{sec:model}}\label{B}
  \subsection[Lemma~\ref{lem:substitution}]{\cref{lem:substitution}\label{B.1}}
  \restate{Lemma}{lem:substitution}{Substitution}{
    If $x:A,\Gamma\vdash t:B$ and $\Gamma\vdash u:A$, then
    the following diagram commutes.
  }
  \[
    \begin{tikzcd}[column sep=1cm]
      \sem\Gamma\ar[r,"(u/x)t"]\ar[d,"\Delta"'] & \sem B \\
      \sem\Gamma\times\sem\Gamma\ar[r,"u\times\Id"'] & \sem A\times\sem\Gamma\ar[u,"t"']
    \end{tikzcd}
  \]
  \begin{proof}
    By induction on $t$.
    \begin{itemize}
      \item Case $t = x$.
	Then $B=A$, $(u/x)t=u$. The diagram is
	\begin{center}
	  \begin{tikzcd}[column sep=1.5cm]
	    \Gamma\ar[r,"u"]\ar[d,"\Delta"'] &  A \\
	    \Gamma\times\Gamma\ar[r,"u\times\Id"'] &  A\times\Gamma\ar[u,"\pi_1"']
	  \end{tikzcd}
	\end{center}
	Straightforward.
      \item Case $t = y$. Then $(u/x)t=y$, $\Gamma=y:B,\Xi$. The diagram is
	\begin{center}
	  \begin{tikzcd}[column sep=1.5cm]
	    B\times \Xi\ar[r,"\pi_{  B}"]\ar[d,"\Delta"'] &   B \\
	    (  B\times \Xi)\times(  B\times \Xi)\ar[r,"u\times\Id"'] &   A\times(  B\times \Xi)\ar[u,"\pi_{  B}"']
	  \end{tikzcd}
	\end{center}
	Straightforward.
      \item Case $t = t_1 \plus t_2$. Then $(u/x)t = (u/x)t_1\plus(u/x)t_2$. The diagram is
	\begin{center}
	  \begin{tikzcd}[column sep=2.5cm,row sep=5mm]
	    \Gamma\ar[rrr,"(u/x)t_1\plus(u/x)t_2"]\ar[ddd,"\Delta"']\ar[rd,"\Delta",dashed,sloped] &[-1.5cm] & &[-1.5cm]   B \\
	    &  \Gamma\times \Gamma\ar[d,"\Delta\times\Delta"',dashed]\ar[r,"(u/x)t_1\times(u/x)t_2",dashed] &   B\times  B\ar[ur,"\sumhat{  B}",dashed,sloped] & \\[5mm]
	    &  \Gamma\times \Gamma\times \Gamma\times \Gamma\ar[r,"u\times\Id\times u\times\Id"',dashed] &   A\times \Gamma\times  A\times \Gamma\ar[u,"t_1\times t_2"',dashed] & \\
	    \Gamma\times \Gamma\ar[rrr,"u\times\Id"']\ar[ru,"\Delta",dashed,sloped] & & &   A\times \Gamma\ar[uuu,"t_1\plus t_2"']\ar[ul,"\Delta",dashed,sloped]
	  \end{tikzcd}
	\end{center}
	\begin{itemize}
	  \item The upper and right diagrams commute by definition.
	  \item The left diagram commutes by coherence.
	  \item The central diagram commutes by induction hypothesis and functoriality of $\times$.
	  \item The bottom diagram commutes by naturality of $\Delta$.
	\end{itemize}
      \item Case $t = \star$. Then $(u/x)t = \star$. The diagram is
	\begin{center}
	  \begin{tikzcd}[column sep=1.5cm]
	    \Gamma\ar[r,"!"]\ar[d,"\Delta"'] &  \{\star\} \\
	    \Gamma\times\Gamma\ar[r,"u\times\Id"'] &  A\times\Gamma\ar[u,"!"']
	  \end{tikzcd}
	\end{center}
	Straightforward, since $\{\star\}$ is terminal.
      \item Case $t = \elimtop(t_1,t_2)$. Then $(u/x)t = \elimtop((u/x)t_1,(u/x)t_2)$. The diagram is
	\begin{center}
	  \begin{tikzcd}[column sep=2.5cm,row sep=5mm]
	    \Gamma\ar[rrr,"{\elimtop((u/x)t_1,(u/x)t_2)}"]\ar[ddd,"\Delta"']\ar[rd,"\Delta",dashed,sloped] &[-1.5cm] & &[-1.5cm]   B \\
	    &  \Gamma\times \Gamma\ar[d,"\Delta\times\Delta"',dashed]\ar[r,"(u/x)t_1\times(u/x)t_2",dashed] &   \{\star\}\times  B\ar[ur,"\pi_2",dashed,sloped] & \\[5mm]
	    &  \Gamma\times \Gamma\times \Gamma\times \Gamma\ar[r,"u\times\Id\times u\times\Id"',dashed] &   A\times \Gamma\times  A\times \Gamma\ar[u,"t_1\times t_2"',dashed] & \\
	    \Gamma\times \Gamma\ar[rrr,"u\times\Id"']\ar[ru,"\Delta",dashed,sloped] & & &   A\times \Gamma\ar[uuu,"{\elimtop(t_1,t_2)}"']\ar[ul,"\Delta",dashed,sloped]
	  \end{tikzcd}
	\end{center}
	\begin{itemize}
	  \item The upper and right diagrams commute by definition.
	  \item The left diagram commutes by coherence.
	  \item The central diagram commutes by induction hypothesis and functoriality of $\times$.
	  \item The bottom diagram commutes by naturality of $\Delta$.
	\end{itemize}
      \item Case $t = \elimbot(t')$. Then $(u/x)t = \elimbot((u/x)t')$. The diagram is
	\begin{center}
	  \begin{tikzcd}[column sep=1.5cm]
	    \Gamma\ar[rd,"(u/x)t'",dashed,sloped]\ar[rr,"{\elimbot((u/x)t')}"]\ar[dd,"\Delta"'] & & B \\
	    & \emptyset\ar[ru,"\emptyset",dashed,sloped] \\
	    \Gamma\times\Gamma\ar[rr,"u\times\Id"'] & & A\times\Gamma\ar[uu,"{\elimbot(t')}"']\ar[ul,"t'",dashed,sloped]
	  \end{tikzcd}
	\end{center}
	\begin{itemize}
	  \item The upper and right diagrams commute by definition.
	  \item The left diagram commutes by the induction hypothesis.
	\end{itemize}
      \item Case $t = \lambda \abstr{y}t'$. Then $B=C\Rightarrow D$, and $(u/x)t = \lambda \abstr{y}(u/x)t'$.
	Notice that since $\Gamma\vdash u:A$, we also have $\Gamma,y:C\vdash u:A$. We write $u_\Gamma = \sem{\Gamma\vdash u:A}$ and $u_{C\times\Gamma} = \sem{C\times\Gamma\vdash u:A} = u_\Gamma\circ\pi_2$.

	The diagram is
	
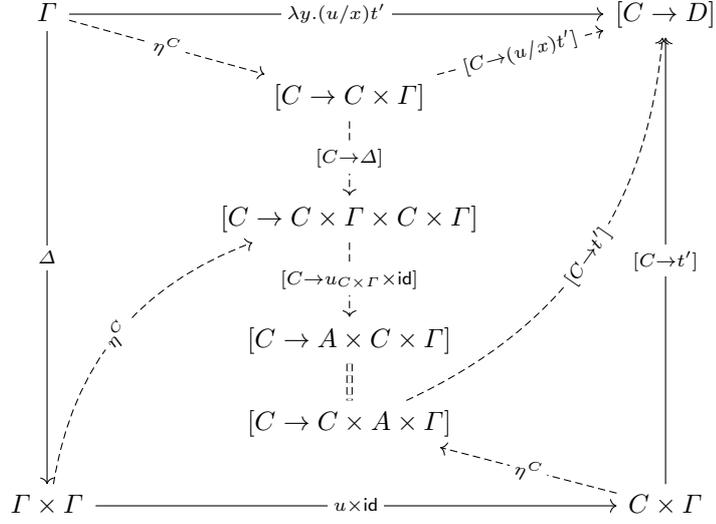
\begin{center}
	  \begin{tikzcd}[column sep=1.5cm]
	    \Gamma\ar[rr,"\lambda\abstr{y}(u/x)t'"]\ar[rd,"\eta^C",dashed,sloped]\ar[ddddd,"\Delta"'] &&  \home{C}{D} \\
	    &\home C{C\times\Gamma}\ar[ru,"\home C{(u/x)t'}",dashed,sloped]\ar[d,"\home C\Delta"',dashed]& \\
	    &\home C{C\times\Gamma\times C\times\Gamma}\ar[dd,"\home C{u_{C\times\Gamma}\times\Id}"',dashed]& \\
	    &\\
	    &\home C{A\times C\times\Gamma}\ar[uuuur,"\home C{t'}",dashed,sloped,pos=0.6]& \\
	    \Gamma\times\Gamma\ar[rr,"u\times\Id"']\ar[uuur,"\eta^C",dashed,sloped,bend left] &&  C\times\Gamma\ar[uuuuu,"\lambda y.t'"']\ar[ul,"\eta^C",dashed,sloped]
	  \end{tikzcd}
	\end{center}
	\begin{itemize}
	  \item The upper and right-bottom diagrams commute by definition.
	  \item The left and bottom diagrams commute by naturality of $\eta^C$.
	  \item The right-top diagram commutes by induction hypothesis and the functoriality of the hom.
	\end{itemize}
      \item Case $t = t_1~t_2$. Then $(u/x)t = (u/x)t_1~(u/x)t_2$. The diagram is
	
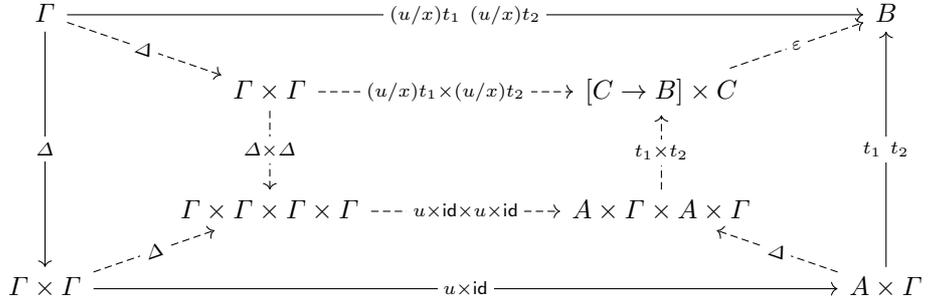
\begin{center}
	  \begin{tikzcd}[column sep=2.5cm,row sep=5mm]
	    \Gamma\ar[rrr,"{(u/x)t_1~(u/x)t_2}"]\ar[ddd,"\Delta"']\ar[rd,"\Delta",dashed,sloped] &[-1.5cm] & &[-1.5cm]   B \\
	    &  \Gamma\times \Gamma\ar[d,"\Delta\times\Delta"',dashed]\ar[r,"(u/x)t_1\times(u/x)t_2",dashed] &   \home CB\times C\ar[ur,"\varepsilon",dashed,sloped] & \\[5mm]
	    &  \Gamma\times \Gamma\times \Gamma\times \Gamma\ar[r,"u\times\Id\times u\times\Id"',dashed] &   A\times \Gamma\times  A\times \Gamma\ar[u,"t_1\times t_2"',dashed] & \\
	    \Gamma\times \Gamma\ar[rrr,"u\times\Id"']\ar[ru,"\Delta",dashed,sloped] & & &   A\times \Gamma\ar[uuu,"{t_1~t_2}"']\ar[ul,"\Delta",dashed,sloped]
	  \end{tikzcd}
	\end{center}
	\begin{itemize}
	  \item The upper and right diagrams commute by definition.
	  \item The left diagram commutes by coherence.
	  \item The central diagram commutes by induction hypothesis and functoriality of $\times$.
	  \item The bottom diagram commutes by naturality of $\Delta$.
	\end{itemize}
      \item Case $t = \pair{t_1}{t_2}$. Then $B=B_1\wedge B_2$ and $(u/x)t = \pair{(u/x)t_1}{(u/x)t_2}$. The diagram is
	
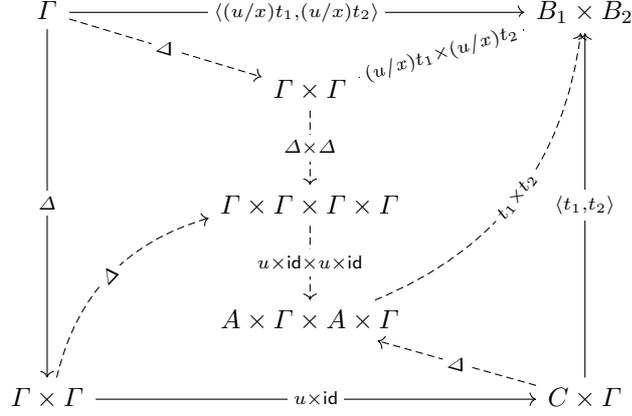
\begin{center}
	  \begin{tikzcd}[column sep=1.5cm,row sep=5mm]
	    \Gamma\ar[rr,"{\pair{(u/x)t_1}{(u/x)t_2}}"]\ar[rd,"\Delta",dashed,sloped]\ar[dddd,"\Delta"'] &&  B_1\times B_2 \\
	    &\Gamma\times\Gamma\ar[ru,"(u/x)t_1\times(u/x)t_2",dashed,sloped]\ar[d,"\Delta\times\Delta"',dashed]& \\[5mm]
	    &{\Gamma\times\Gamma\times\Gamma\times\Gamma}\ar[d,"u\times\Id\times u\times\Id"',dashed]& \\[5mm]
	    &A\times\Gamma\times A\times\Gamma\ar[uuur,"t_1\times t_2",dashed,sloped]& \\
	    \Gamma\times\Gamma\ar[rr,"u\times\Id"']\ar[uur,"\Delta",dashed,sloped,bend left] &&  C\times\Gamma\ar[uuuu,"\pair{t_1}{t_2}"']\ar[ul,"\Delta",dashed,sloped]
	  \end{tikzcd}
	\end{center}
	\begin{itemize}
	  \item The upper and right-bottom diagrams commute by definition.
	  \item The left and bottom diagrams commute by naturality of $\Delta$.
	  \item The right-top diagram commutes by induction hypothesis and the functoriality of the $\times$.
	\end{itemize}
      \item Case $t = \pi_1(t')$. Then $(u/x)t = \pi_1((u/x)t')$. The diagram is
	
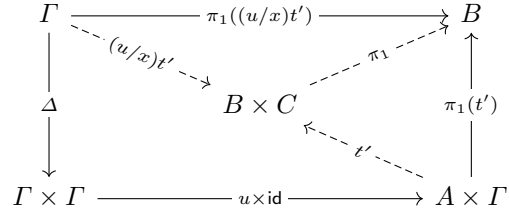
\begin{center}
	  \begin{tikzcd}[column sep=1.5cm]
	    \Gamma\ar[rd,"(u/x)t'",dashed,sloped]\ar[rr,"{\pi_1((u/x)t')}"]\ar[dd,"\Delta"'] & & B \\
	    & B\times C\ar[ru,"\pi_1",dashed,sloped] \\
	    \Gamma\times\Gamma\ar[rr,"u\times\Id"'] & & A\times\Gamma\ar[uu,"{\pi_1(t')}"']\ar[ul,"t'",dashed,sloped]
	  \end{tikzcd}
	\end{center}
	\begin{itemize}
	  \item The upper and right diagrams commute by definition.
	  \item The left diagram commutes by the induction hypothesis.
	\end{itemize}
      \item Case $t = \pi_2(t')$. This case is analogous to the previous one.
      \item Case $t = \inl(t')$. Then $B=B_1\cp B_2$ and $(u/x)t = \inl((u/x)t')$. The diagram is
	
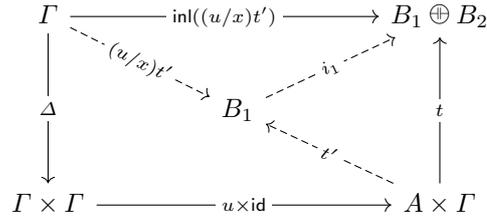
\begin{center}
	  \begin{tikzcd}[column sep=1.5cm]
	    \Gamma\ar[rd,"(u/x)t'",dashed,sloped]\ar[rr,"{\inl((u/x)t')}"]\ar[dd,"\Delta"'] & & B_1\cp B_2 \\
	    & B_1\ar[ru,"i_1",dashed,sloped] \\
	    \Gamma\times\Gamma\ar[rr,"u\times\Id"'] & & A\times\Gamma\ar[uu,"{t}"']\ar[ul,"t'",dashed,sloped]
	  \end{tikzcd}
	\end{center}
	\begin{itemize}
	  \item The upper and right diagrams commute by definition.
	  \item The left diagram commutes by the induction hypothesis.
	\end{itemize}
      \item Case $t = \inr(t')$. This case is analogous to the previous one.
      \item Case $t = \elimor(t_1,\abstr{y}t_2,\abstr{z}t_3)$. Then $(u/x)t$ is equal to $\elimor((u/x)t_1,\abstr{y}(u/x)t_2,\abstr{z}(u/x)t_3)$. The diagram is given in \cref{fig:substitutionelimor}.
	\begin{figure}[th!]
	  \[
	    \begin{tikzcd}[labels=description,column sep=12mm,row sep=1.2cm,
		execute at end picture={
		  \path (\tikzcdmatrixname-1-1) -- (\tikzcdmatrixname-1-3) coordinate[pos=0.4](aux1)
		  (\tikzcdmatrixname-2-1) -- (\tikzcdmatrixname-2-3) coordinate[pos=0.4](aux2)
		  (\tikzcdmatrixname-2-1) -- (\tikzcdmatrixname-8-1) coordinate[pos=0.5](aux3)
		  (aux3) -- (\tikzcdmatrixname-5-1) node[red,midway,xshift=-5mm]{\small ({1})}
		  (aux1) -- (aux2) node[red,midway]{\small ({2})}
		  (aux2) -- (\tikzcdmatrixname-3-2) node[red,midway,xshift=-2mm]{\small ({3})}
		  (aux3) -- (\tikzcdmatrixname-4-2) node[red,midway]{\small ({4})}
		  (aux3) -- (\tikzcdmatrixname-6-2) node[red,midway]{\small ({5})}
		  (\tikzcdmatrixname-5-1) -- (\tikzcdmatrixname-8-1) coordinate[midway](aux4)
		  (aux4) -- (\tikzcdmatrixname-8-2) node[red,midway]{\small ({6})}
		  (\tikzcdmatrixname-8-1) -- (\tikzcdmatrixname-9-3) node[red,pos=0.4]{\small ({7})}
		  (\tikzcdmatrixname-3-2) -- (\tikzcdmatrixname-3-3) node[red,midway]{\small ({8})}
		  (\tikzcdmatrixname-3-2) -- (\tikzcdmatrixname-5-3) node[red,pos=0.5,xshift=3mm]{\small ({9})}
		  (\tikzcdmatrixname-3-2) -- (\tikzcdmatrixname-5-3) node[red,pos=0.7,xshift=-1.4cm]{\small ({10})}
		  (\tikzcdmatrixname-4-2) -- (\tikzcdmatrixname-8-3) node[red,midway]{\small ({11})}
		  (\tikzcdmatrixname-6-2) -- (\tikzcdmatrixname-7-3) node[red,pos=0.3]{\small ({12})}
		  (\tikzcdmatrixname-6-2) -- (\tikzcdmatrixname-7-2) node[red,midway]{\small ({13})}
		  (\tikzcdmatrixname-7-2) -- (\tikzcdmatrixname-9-3) coordinate[midway](aux5)
		  (aux5) -- (\tikzcdmatrixname-8-3) node[red,pos=0.7]{\small ({14})}
		  (aux5) -- (\tikzcdmatrixname-8-2) node[red,midway]{\small ({15})}
		  (\tikzcdmatrixname-2-3) -- (\tikzcdmatrixname-3-3) node[red,midway,xshift=1cm]{\small ({16})}
		  (\tikzcdmatrixname-3-3) -- (\tikzcdmatrixname-5-3) node[red,midway,xshift=1.1cm]{\small ({17})}
		  (\tikzcdmatrixname-7-3) -- (\tikzcdmatrixname-8-3) node[red,midway,xshift=1cm]{\small ({18})}
		  ;
		}
	      ]
	      \Gamma
	      \arrow["\Delta", from=1-1, to=9-1,
		rounded corners,
		to path={[pos=0.75]
		  -- ([xshift=-5mm]\tikztostart.west)
		  -| ([xshift=-1cm]\tikztotarget.west)\tikztonodes
		-- (\tikztotarget.west)}
	      ]
	      &[2mm] &[-1mm] B \\
	      {(C\cp D)\times\Gamma} && {(C\times\Gamma)\cp(D\times\Gamma)} \\
	      & {
		\begin{array}{c}
		  (C\times(\Gamma\times\Gamma)) \\
		  \cp                           \\
		  (D\times(\Gamma\times\Gamma))
		\end{array}
	      } & {
		\begin{array}{c}
		  ((C\times\Gamma)\times(C\times\Gamma)) \\
		  \cp                                    \\
		  ((D\times\Gamma)\times(D\times\Gamma))
		\end{array}
	      }
	      \arrow["{(u_{C\times\Gamma}\times\Id)\cplabel(u_{C\times\Gamma}\times\Id)}", sloped,out=-45,in=45,looseness=0.9, dashed, from=3-3, to=5-3]
	      \\
	      & {(C\cp D)\times(\Gamma\times\Gamma)} & {
		\begin{array}{c}
		  (\Gamma\times(C\times\Gamma)) \\
		  \cp                           \\
		  (\Gamma\times(D\times\Gamma))
		\end{array}
	      } \\
	      {(A\times\Gamma)\times\Gamma} & {(A\times\Gamma)\times(\Gamma\times\Gamma)} & {
		\begin{array}{c}
		  (A\times C\times\Gamma) \\
		  \cp                     \\
		  (B\times D\times\Gamma)
		\end{array}
	      } \\
	      & {(A\times\Gamma)\times(\Gamma\times\Gamma)} & {
		\begin{array}{c}
		  (C\times A\times\Gamma) \\
		  \cp                     \\
		  (D\times A\times\Gamma)
		\end{array}
	      }
	      \arrow["{\coproducto{t_2}{t_3}}",sloped, from=6-3, to=1-3,
		rounded corners,dashed,
		to path={[pos=0.75]
		  -- ([xshift=1mm]\tikztostart.east)
		  -| ([xshift=1.6cm,yshift=-2mm]\tikztotarget.east)\tikztonodes
		-- (\tikztotarget.east)}
	      ]
	      \\
	      & {(A\times A)\times(\Gamma\times\Gamma)} & {(C\cp D)\times (A\times\Gamma)} \\
	      {(\Gamma\times\Gamma)\times\Gamma} & {(A\times A)\times\Gamma} & {(A\times\Gamma)\times (A\times\Gamma)} \\
	      \Gamma\times\Gamma
	      \arrow["{(u/x)t_1\times\Id}", from=9-1, to=2-1,
		rounded corners, sloped,dashed,
		to path={[pos=0.75]
		  -- ([xshift=-1mm,yshift=1mm]\tikztostart.west)
		  -| ([xshift=-2mm]\tikztotarget.west)\tikztonodes
		-- (\tikztotarget.west)}
	      ]
	      && A\times\Gamma
	      \arrow["{\elimor(t_1,\abstr{y}t_2,\abstr{z}t_3)}",sloped, from=9-3, to=1-3,
		rounded corners,
		to path={[pos=0.75]
		  -- ([xshift=3mm]\tikztostart.east)
		  -| ([xshift=1.8cm]\tikztotarget.east)\tikztonodes
		-- (\tikztotarget.east)}
	      ]
	      \arrow["{(\Id\times u_\Gamma)\times(\Id\times\Id)}", out=215,in=160,sloped, dashed, from=5-2, to=7-2]
	      \arrow["{\elimor((u/x)t_1,\abstr{y}(u/x)t_2,\abstr{z}(u/x)t_3)}", from=1-1, to=1-3]
	      \arrow["d"', from=2-1, to=2-3,dashed]
	      \arrow["\Id\times\Delta", dashed, from=2-1, to=4-2,sloped]
	      \arrow["{\coproducto{(u/x)t_2}{(u/x)t_3}}"', from=2-3, to=1-3,dashed]
	      \arrow["{{\Delta}\cplabel{\Delta}}"', dashed, from=2-3, to=3-3]
	      \arrow["{(\sigma\times\Id)\cplabel(\sigma\times\Id)}", dashed, to=3-2, from=4-3,sloped]
	      \arrow["{(\Id\times(u_\Gamma\times\Id))\cplabel (\Id\times(u_\Gamma\times\Id))}", dashed, from=3-2, to=6-3,sloped,out=-45,in=135]
	      \arrow["{(\pi_2\times\Id)\cplabel(\pi_2\times\Id)}"', dashed, from=3-3, to=4-3]
	      \arrow["d", from=4-2, to=3-2,dashed]
	      \arrow["{\Id\times(u\times\Id)}", dashed, from=4-2, to=7-3,sloped]
	      \arrow["{(u_{\Gamma}\times\Id)\cplabel(u_{\Gamma}\times\Id)}"', dashed, from=4-3, to=5-3]
	      \arrow["{t_1\times\Id}"', dashed, from=5-1, to=2-1]
	      \arrow["\Id\times\Delta", dashed, from=5-1, to=5-2]
	      \arrow["{(\Id\times u_\Gamma)\times\Id}", dashed, from=5-1, to=8-2,sloped,bend right=10]
	      \arrow["{t_1\times\Id}", dashed, from=5-2, to=4-2]
	      \arrow["{\Id\times\sigma\times\Id}", dashed, from=5-2, to=6-2]
	      \arrow["{\Id\times (u_\Gamma\times\Id)}", dashed, from=5-2, to=8-3,sloped,out=-45,in=140]
	      \arrow["(\sigma\times\Id)\cplabel(\sigma\times\Id)", dashed, to=6-3, from=5-3]
	      \arrow["{\Id\times(u_\Gamma\times\Id)}", dashed, from=6-2, to=8-3,sloped]
	      \arrow["\Id\times\sigma\times\Id", dashed, from=7-2, to=8-3,sloped]
	      \arrow["d"', from=7-3, to=6-3,dashed]
	      \arrow["{(u_\Gamma\times\Id)\times\Id}"', dashed, from=8-1, to=5-1]
	      \arrow["{(u_\Gamma\times u_\Gamma)\times\Id}", dashed, from=8-1, to=8-2]
	      \arrow["\Id\times\Delta", dashed, from=8-2, to=7-2]
	      \arrow["\Delta\times\Delta", dashed, from=9-3, to=7-2,sloped]
	      \arrow["{t_1\times\Id}"', from=8-3, to=7-3,dashed]
	      \arrow["\Delta\times\Id"', dashed, from=9-1, to=8-1]
	      \arrow["{u_\Gamma\times\Id}"', from=9-1, to=9-3]
	      \arrow["\Delta\times\Id", dashed, from=9-3, to=8-2,sloped]
	      \arrow["\Delta"', from=9-3, to=8-3,dashed]
	      \arrow["(\Id\times\Delta)\cplabel(\Id\times\Delta)", dashed, from=2-3, to=3-2,sloped]
	    \end{tikzcd}
	  \]
	  \caption{Substitution for $\elimor$. Proof of \cref{lem:substitution} in \appref{B.1}.}
	  \label{fig:substitutionelimor}
	\end{figure}
	The commutation of each subdiagram is justified as follows:
	\begin{itemize}
	  \item Diagram (1): Induction hypothesis and functoriality of $\times$.
	  \item Diagrams (2) and (18): By definition.
	  \item Diagrams (3) and (10): Naturality of $d$.
	  \item Diagrams (4), (5), (6), (11), and (15): functoriality of $\times$.
	  \item Diagram (7): Naturality of $\Delta$ and functoriality of $\times$.
	  \item Diagram (8): $(\pi_2\times\Id_{C\times\Gamma})\circ\Delta = (\sigma\times\Id_{\Gamma\times\Gamma})\circ(\Id_C\times\Delta)$, and functoriality of $\cp$ (\cref{lem:cpbifunctorial}).
	  \item Diagram (9): Naturality of $\sigma$ and functoriality of $\times$ and $\cp$ (\cref{lem:cpbifunctorial}).
	  \item Diagram (12): This diagram commutes only when precomposed with 
	    \[
	      A\times\Gamma\xlra{(\Id_{A\times\Gamma})\circ(\Id_A\times\Delta)} A\times\Gamma\times\Gamma\times\Gamma
	    \]
	    which is the case here, since all the $\Gamma$ that appear in this diagram starts from the same point.

	    Hence, with this precomposition, the diagram can be proven by evaluation. On the top side we have
	    \[
	      (a,g)\mapsto ((a,g),(g,g))
	      \mapsto ((a,g),(u(g),g))
	    \]
	    On the bottom side we have
	    \[
	      (a,g)\mapsto ((a,g),(g,g))
	      \mapsto ((a,g),(g,g))
	      \mapsto ((a,g),(u(g),g))
	    \]
	  \item Diagram (13): Naturality of $\sigma$.
	  \item Diagram (14): Straightforward.
	  \item Diagram (16): Induction hypothesis and functoriality of $\cp$ (\cref{lem:cpbifunctorial}).
	  \item Diagram (17): Weakening.
	    \qedhere
	\end{itemize}
    \end{itemize}
  \end{proof}

  \subsection[Theorem~\ref{thm:soundness}]{\cref{thm:soundness}}\label{B.2}
  \restate{Theorem}{thm:soundness}{Soundness}{
    If $t\lra r$ and $\Gamma\vdash t:A$, then
    $\sem{\Gamma\vdash t:A} = \sem{\Gamma\vdash r:A}$.
  }
  \begin{proof}
    By induction on the relation $\longrightarrow$.
    \begin{itemize}
      \item Rule $\elimtop(\star,t)  \longrightarrow t  $. The commuting diagram is
	\[
	  \begin{tikzcd}[column sep=1cm,row sep=2mm]
	    \Gamma &&& A \\
	    & {g} & {t(g)} \\[3mm]
	    & {(g,g)} & {(\star,t(g))} \\
	    \Gamma\times\Gamma &&& {\{\star\}\times A}
	    \arrow["t", from=1-1, to=1-4]
	    \arrow["\Delta"', from=1-1, to=4-1]
	    \arrow[dotted, blue, no head, from=2-2, to=1-1]
	    \arrow[maps to,  from=2-2, to=2-3]
	    \arrow[maps to,  from=2-2, to=3-2]
	    \arrow[dotted, blue, no head, from=2-3, to=1-4]
	    \arrow[maps to,  from=3-2, to=3-3]
	    \arrow[dotted, blue, no head, from=3-2, to=4-1]
	    \arrow[maps to,  from=3-3, to=2-3]
	    \arrow[dotted, blue, no head, from=3-3, to=4-4]
	    \arrow["{{!}\times t}"', from=4-1, to=4-4]
	    \arrow["{\pi_2}"', from=4-4, to=1-4]
	  \end{tikzcd}
	\]
      \item Rule $(\lambda \abstr{x}t)~u  \longrightarrow (u/x)t  $. The commuting diagram is
	\[
	  \begin{tikzcd}[column sep=2cm,row sep=5mm]
	    \Gamma\ar[rr,"(u/x)t"]\ar[dd,"\Delta"'] & & A\\
	    & B\times\Gamma\ar[ur,"t",sloped,dashed]\ar[d,"\Id\times\eta^B"',dashed,bend right] & \\[1cm]
	    \Gamma\times\Gamma\ar[ru,"u\times\Id",sloped,dashed]\ar[r,"u\times\eta^B"'] & B\times\home B{B\times\Gamma}\ar[u,"\varepsilon"',dashed,bend right]\ar[r,"\Id\times\home B{t}"'] & B\times\home BA\ar[uu,"\varepsilon"']
	  \end{tikzcd}
	\]
	\begin{itemize}
	  \item The upper diagram commutes by \cref{lem:substitution}
	  \item The bottom-left diagram commutes by functoriality of $\times$.
	  \item The bottom-middle diagram commutes by axiom of the adjunction.
	  \item The bottom-right diagram commutes by naturality of $\varepsilon$.
	\end{itemize}
      \item Rule $\pi_1\pair{t}{u}  \longrightarrow t  $. The commuting diagram is
	\[
	  \begin{tikzcd}[column sep=1cm,row sep=2mm]
	    \Gamma &&& A \\
	    & {g} & {t(g)} \\[3mm]
	    & {(g,g)} & {(t(g),u(g))} \\
	    \Gamma\times\Gamma &&& {A\times B}
	    \arrow["t", from=1-1, to=1-4]
	    \arrow["\Delta"', from=1-1, to=4-1]
	    \arrow[dotted, blue, no head, from=2-2, to=1-1]
	    \arrow[maps to,  from=2-2, to=2-3]
	    \arrow[maps to,  from=2-2, to=3-2]
	    \arrow[dotted, blue, no head, from=2-3, to=1-4]
	    \arrow[maps to,  from=3-2, to=3-3]
	    \arrow[dotted, blue, no head, from=3-2, to=4-1]
	    \arrow[maps to,  from=3-3, to=2-3]
	    \arrow[dotted, blue, no head, from=3-3, to=4-4]
	    \arrow["{t\times g}"', from=4-1, to=4-4]
	    \arrow["{\pi_1}"', from=4-4, to=1-4]
	  \end{tikzcd}
	\]
      \item Rule $\pi_2\pair{t}{u}  \longrightarrow u  $. This case is analogous to the previous one.
      \item Rule $\elimor(\inl(t),\abstr{x}v,\abstr{y}w) \longrightarrow (t/x)v  $.
	The commuting diagram is
	\[
	  \begin{tikzcd}[column sep=2cm]
	    \Gamma\ar[r,"(t/x)v"]\ar[d,"\Delta"'] & A \\
	    \Gamma\times\Gamma \ar[d,"t\times\Id"'] & (B\times\Gamma)\cp(C\times\Gamma)\ar[u,"\coproducto{v}{w}"']\\
	    B\times\Gamma\ar[uur,"v",sloped,dashed]\ar[r,"i_1\times\Id"']\ar[ur,"i_1",sloped,dashed] & (B\cp C)\times\Gamma\ar[u,"d"']
	  \end{tikzcd}
	\]
	\begin{itemize}
	  \item The upper diagram commutes by \cref{lem:substitution}.
	  \item The middle diagram commutes by \cref{lem:coprod}.
	  \item The bottom diagram commutes since for any $b\in B$ and $g\in\Gamma$ we have
	    \[
	      d(i_1(b),g)=d((b,0),g)=((b,g),0) = i_1(b,g)
	    \]
	\end{itemize}
      \item Rule $\elimor(\inr(u),\abstr{x}v,\abstr{y}w) \longrightarrow (u/y)w  $.
	This case is analogous to the previous one.
      \item Rule $\elimor(\inl(t)\plus\inr(u),\abstr{x}v,\abstr{y}w)\longrightarrow (t/x)v\plus (u/y)w$. 
	The commuting diagram is
	\[
	  \begin{tikzcd}[column sep=1cm,row sep=5mm]
	    \Gamma \\
	    {\Gamma\times\Gamma} &[-5mm] &[-4mm]&[-1.4cm] {(\Gamma\times\Gamma)\times\Gamma} \\
	    & {(\Gamma\times\Gamma)\times(\Gamma\times\Gamma)} \\
	    & {(\Gamma\times\Gamma)\times(\Gamma\times\Gamma)} \\
	    && {(A\times B)\times(\Gamma\times\Gamma)} \\
	    {C\times C} & {(A\times\Gamma)\times (B\times\Gamma)} && {(A\times B)\times\Gamma} \\
	    C &&& {(A\cp B)\times(A\cp B)\times\Gamma} \\
	    {(A\times\Gamma)\cp(B\times\Gamma)} &&& {(A\cp B)\times\Gamma}
	    \arrow["\Delta"', from=1-1, to=2-1]
	    \arrow["{\Delta\times\Id}", from=2-1, to=2-4]
	    \arrow["{\Delta\times\Delta}", dashed, from=2-1, to=3-2,sloped]
	    \arrow["{(t/x)v\times (u/x)w}"',sloped, from=2-1, to=6-1]
	    \arrow["{\Id\times\Delta}", dashed, from=2-4, to=3-2,sloped]
	    \arrow["{(t\times u)\times\Id}", from=2-4, to=6-4]
	    \arrow["{\Id\times\sigma\times\Id}", dashed, from=3-2, to=4-2]
	    \arrow["{(t\times\Id)\times (u\times\Id)}"', out=200,in=160, dashed, from=3-2, to=6-2]
	    \arrow["{(t\times\Id)\times (u\times\Id)}", dashed, from=4-2, to=6-2,pos=0.3]
	    \arrow["{\Id\times\sigma\times\Id}", dashed, from=5-3, to=6-2,sloped]
	    \arrow["{\sumhat C}"', from=6-1, to=7-1]
	    \arrow["{v\times w}", dashed, from=6-2, to=6-1]
	    \arrow["{\Id\times\Delta}"', dashed, from=6-4, to=5-3]
	    \arrow["\delta", dashed, from=6-4, to=6-2]
	    \arrow["{(i_1\times i_2)\times\Id}", from=6-4, to=7-4]
	    \arrow["{\sumcp\times\Id}", from=7-4, to=8-4]
	    \arrow["{\coproducto{v}{w}}", from=8-1, to=7-1]
	    \arrow["d", from=8-4, to=8-1]
	    \arrow["\inclusion", from=6-2, to=8-1,sloped,dashed]
	    \arrow["\inclusion"', from=6-4, to=8-4,dashed,out=185,in=175,looseness=1.4]
	  \end{tikzcd}
	\]
	\begin{itemize}
	  \item The top diagram commutes by the functoriality of $\times$.
	  \item The upper-left diagram commutes by \cref{lem:substitution} and the functoriality of $\times$.
	  \item The upper-middle diagram commutes by coherence, when precomposed with $\Delta$s.
	  \item The upper-right-top diagram commutes by the naturality of $(\Id\times\sigma\times\Id)\circ(\Id\times\Delta)$.
	  \item The upper-right-bottom diagram is the definition of the map $\delta$.
	  \item The bottom-left diagram commutes by definition of $\coproducto vw$.
	  \item The bottom-right diagram commutes by definition of $\sumcp$.
	  \item The bottom-middle diagram commutes since $\delta$ is a restriction of $d$.
	\end{itemize}

      \item Rule ${\star} \plus \star  \longrightarrow \star  $. The commuting diagram is
	\[
	  \begin{tikzcd}[column sep=1.5cm]
	    \Gamma\ar[r,"{!}"]\ar[d,"\Delta"'] & \{\star\} \\
	    \Gamma\times\Gamma\ar[r,"{!}\times{!}"'] & \{\star\}\times\{\star\}\ar[u,"\sumhat{\{\star\}}"']
	  \end{tikzcd}
	\]
	Straightforward.
      \item Rule $(\lambda \abstr{x}t) \plus (\lambda \abstr{x}u) \longrightarrow \lambda \abstr{x}(t \plus u) $. The commuting diagram is
	\[
	  \begin{tikzcd}[column sep=1.5cm]
	    \Gamma & {\home A{A\times\Gamma}} \\
	    \Gamma\times\Gamma & {\home A{(A\times\Gamma)\times(A\times\Gamma)}} \\
	    {\home A{A\times\Gamma}\times\home A{A\times\Gamma}} & {\home A{B\times B}} \\
	    {\home AB\times\home AB} & {\home AB}
	    \arrow["{\eta^A}"', from=1-1, to=1-2]
	    \arrow["\Delta"', from=1-1, to=2-1]
	    \arrow["{\home A\Delta}", from=1-2, to=2-2]
	    \arrow["\Delta", from=1-2, to=3-1,dashed,sloped]
	    \arrow["{\eta^A\times\eta^A}"', from=2-1, to=3-1]
	    \arrow["{\home A{t\times u}}", from=2-2, to=3-2]
	    \arrow["{\home At\times\home Au}"', from=3-1, to=4-1]
	    \arrow["{\home A{\sumhat B}}", from=3-2, to=4-2]
	    \arrow["\sumhat{\home AB}"', from=4-1, to=4-2]
	  \end{tikzcd}
	\]
	\begin{itemize}
	  \item The top diagram commutes by naturality of $\Delta$.
	  \item The bottom diagram commutes since
	    \[
	      \sumhat{\home AB}\circ(t\times u)\circ\Delta\circ f = (t\circ f)\sumhat B(u\circ f)
	    \]
	\end{itemize}
      \item Rule $\pair{t_1}{u_1} \plus \pair{t_2}{u_2}  \longrightarrow \pair{t_1 \plus t_2}{u_1 \plus u_2} $. The commuting diagram is
	\[
	  \begin{tikzcd}[column sep=2.5cm]
	    \Gamma\times\Gamma\ar[d,"\Delta\times\Delta"'] & \Gamma\ar[l,"\Delta"']\\
	    (\Gamma\times\Gamma)\times(\Gamma\times\Gamma)\ar[r,"(t_1\times t_2)\times(u_1\times u_2)"]\ar[d,"(t_1\times u_1)\times(t_2\times u_2)"'] & (A\times A)\times(B\times B)\ar[d,"\sumhat A\times\sumhat B"]\\
	    (A\times B)\times(A\times B)\ar[r,"\sumhat{A\times B}"']\ar[ur,"\simeq",sloped,dashed] & A\times B
	  \end{tikzcd}
	\]
	\begin{itemize}
	  \item The top diagram commutes by coherence.
	  \item The bottom diagram commutes since
	    \[
	      (a,b)\sumhat{A\times B}(a',b') = (a\sumhat A a',b\sumhat B b')
	    \]
	\end{itemize}
      \item Rule $\inl(t_1) \plus \inl(t_2)  \longrightarrow \inl(t_1 \plus t_2) $. The commuting diagram is
	\[
	  \begin{tikzcd}[column sep=1.2cm]
	    \Gamma\times\Gamma\ar[d,"t_1\times t_2"'] && \Gamma\ar[ll,"\Delta"']\\
	    A\times A\ar[r,"\sumhat A"]\ar[dr,"i_1\times i_1"',sloped] & A\ar[r,"i_1"] & A\cp B\\
	    & (A\cp B)\times (A\cp B)\ar[ur,"\sumcp"',sloped]
	  \end{tikzcd}
	\]
	The diagram commutes by \cref{lem:injhomomorphism}.

      \item Rule $\inl(t_1) \plus (\inl(t_2)\plus\inr(u_1))  \longrightarrow \inl(t_1\plus t_2)\plus\inr(u_1) $. The commuting diagram is
	\[
	  \begin{tikzcd}[column sep=1.2cm]
	    { \Gamma\times\Gamma} & \Gamma \\
	    {\Gamma\times(\Gamma\times\Gamma)} && {(\Gamma\times\Gamma)\times\Gamma} \\
	    {A\times (A\times B)} && {(A\times A)\times B} \\
	    {(A\cp B)\times(((A\cp B)\times (A\cp B)))} && {A\times B} \\
	    {(A\cp B)\times(A\cp B)} & {A\cp B} & {(A\cp B)\times(A\cp B)}
	    \arrow["{\Id\times\Delta}"', from=1-1, to=2-1] 
	    \arrow["{\Delta\times\Id}", from=1-1, to=2-3,sloped] 
	    \arrow["\Delta"', from=1-2, to=1-1] 
	    \arrow["{t_1\times (t_2\times u_1)}"', from=2-1, to=3-1] 
	    \arrow["\alpha"', dashed, from=2-3, to=2-1] 
	    \arrow["{(t_1\times t_2)\times u_1}", from=2-3, to=3-3]
	    \arrow["{i_1\times (i_1\times i_2)}"', from=3-1, to=4-1]
	    \arrow["{i_1\times\inclusion}", dashed, from=3-1, to=5-1,out=-10,in=10,pos=0.3,looseness=2] 
	    \arrow["\alpha"', dashed, from=3-3, to=3-1] 
	    \arrow["{\sumhat A\times\Id}", from=3-3, to=4-3] 
	    \arrow["{\Id\times\sumcp}", from=4-1, to=5-1]
	    \arrow["\inclusion", dashed, from=4-3, to=5-2,sloped]
	    \arrow["{i_1\times i_2}", from=4-3, to=5-3] 
	    \arrow["\sumcp"', from=5-1, to=5-2] 
	    \arrow["\sumcp", from=5-3, to=5-2]
	  \end{tikzcd}
	\]

	\begin{itemize}
	  \item The top diagram commutes because $\alpha\circ(\Delta\times\Id)\circ\Delta = (\Id\circ\Delta)\circ\Delta$. Notice that it only commutes when precomposed with $\Delta$.
	  \item The middle diagram commutes by naturality of $\alpha$.
	  \item The bottom-left diagram commutes by functoriality of $\times$ and \cref{lem:inclusion}.
	  \item The bottom-right diagram commutes by \cref{lem:inclusion}.
	  \item The bottom-center diagram commutes since $(a,0)\sumcp (a',b) = (a\sumhat A a',b)$.
	\end{itemize}

      \item $\inr(u_1)\parallel\inl(t_1)\longrightarrow\inl(t_1)\parallel\inr(u_1)$. The commuting diagram is 
	\[
	  \begin{tikzcd}[column sep=1.5cm]
	    \Gamma \\
	    { \Gamma\times\Gamma} & {A\times B} & {(A\cp B)\times(A\cp B)} \\
	    {B\times A} & {(A\cp B)\times(A\cp B)} & {(A\cp B)}
	    \arrow["\Delta", from=1-1, to=2-1]
	    \arrow["{t_1\times u_1}", from=2-1, to=2-2]
	    \arrow["{u_1\times t_1}"', from=2-1, to=3-1]
	    \arrow["{i_1\times i_2}", from=2-2, to=2-3]
	    \arrow["\simeq", dashed, from=2-2, to=3-1,sloped]
	    \arrow["\sumcp", from=2-3, to=3-3]
	    \arrow["{i_2\times i_1}"', from=3-1, to=3-2]
	    \arrow["\sumcp"', from=3-2, to=3-3]
	  \end{tikzcd}
	\]
	\begin{itemize}
	  \item The left diagram commutes by coherence.
	  \item The right diagram commutes since 
	    \begin{align*}
	      \sumcp\circ(i_2\times i_1)(b,a)
	      &=(b,1)\sumcp (a,0) \\
	      &= (a,b)\\
	      &= (a,0)\sumcp (b,1)
	      = \sumcp\circ(i_1\times i_2)(a,b)
	    \end{align*}
	\end{itemize}

      \item Rule $\inr(u_1) \plus \inr(u_2)  \longrightarrow \inr(u_1 \plus u_2)$. This case is analogous to the case of $\inl(t_1)\plus\inl(t_2)\longrightarrow \inl(t_1\plus t_2)$.

      \item Rule $\inr(u_1) \plus (\inl(t_1)\plus\inr(u_2))  \longrightarrow \inl(t_1)\plus\inr(u_1 \plus u_2) $. The commuting diagram is
	\[
	  \begin{tikzcd}[column sep=1cm,row sep=1cm]
	    { \Gamma\times\Gamma} & \Gamma \\
	    {\Gamma\times(\Gamma\times\Gamma)} &{B\times (A\times B)}& {(A\cp B)\times((A\cp B)\times (A\cp B))} \\
	    {A\times(B\times B)} && {(A\cp B)\times (A\cp B)} \\
	    {A\times B} &{(A\cp B)\times(A\cp B)}& {A\cp B}
	    \arrow["{\Id\times \Delta}"', from=1-1, to=2-1]
	    \arrow["\Delta"', from=1-2, to=1-1]
	    \arrow["{u_1\times (t_1\times u_2)}",sloped, from=2-1, to=2-2]
	    \arrow["{t_1\times(u_1\times u_2)}"', from=2-1, to=3-1]
	    \arrow["{\Id\times\sumcp}", from=2-3, to=3-3]
	    \arrow["{i_2\times(i_1\times i_2)}",sloped, from=2-2, to=2-3]
	    \arrow["{i_2\times\inclusion}"',sloped, dashed, from=2-2, to=3-3]
	    \arrow["{\simeq}"',sloped, dashed, from=3-1, to=2-2]
	    \arrow["{\Id\times\sumhat B}"', from=3-1, to=4-1]
	    \arrow["\sumcp", from=3-3, to=4-3]
	    \arrow["\inclusion", dashed, from=4-1, to=4-3,bend left=15,sloped]
	    \arrow["{i_1\times i_2}"', from=4-1, to=4-2,sloped]
	    \arrow["\sumcp"', from=4-2, to=4-3,sloped]
	  \end{tikzcd}
	\]
	\begin{itemize}
	  \item The left diagram commutes by coherence.
	  \item The right diagram commutes by the functoriality of $\times$ and \cref{lem:inclusion}.
	  \item The bottom diagram commutes by \cref{lem:inclusion}.
	  \item The middle diagram commutes since
	    \[
	      \sumcp\circ (i_2\times\inclusion) (b_1,(a,b_2))
	      = (b_1,1)\sumcp (a,b_2)
	      = (a,b_1\sumhat B b_2)
	    \]
	\end{itemize}

      \item Rule $(\inl(t_1)\plus\inr(u_1)) \plus \inl(t_1)  \longrightarrow \inl(t_1\plus t_2)\plus\inr(u_1) $. This case is analogous to the case of $\inr(u_1)\plus(\inl(t_1)\plus\inr(u_2))\longrightarrow\inl(t_1)\plus\inr(u_1\plus u_2)$, where the commutation of the middle diagram is given by
	\[
	  \sumcp\circ(\inclusion\times i_1)((a_1,b),a_2)
	  = (a_1,b)\sumcp (a_2,0)
	  = (a_1\sumhat A a_2,b)
	\]

      \item Rule $(\inl(t_1)\plus\inr(u_1)) \plus \inr(u_2)  \longrightarrow \inl(t_1)\plus\inr(u_1\plus u_2) $. This case is analogous to the case $\inl(t_1)\plus(\inl(t_2)\plus\inr(u_1))\longrightarrow\inl(t_1\plus t_2)\plus\inr(u_1)$, where the bottom-center diagram commutes since 
	$(a,b)\sumcp(b',1) = (a,b\sumhat B b')$.

      \item Rule $(\inl(t_1)\plus\inr(u_1)) \plus (\inl(t_2)\plus\inr(u_2))  \longrightarrow \inl(t_1\plus t_2)\plus\inr(u_1\plus u_2) $. The commuting diagram is
	\[
	  \begin{tikzcd}[column sep=1cm]
	    { \Gamma\times\Gamma} & \Gamma \\
	    {(\Gamma\times\Gamma)\times(\Gamma\times\Gamma)} \\
	    {(A\times B)\times (A\times B)} && {(A\times A)\times (B\times B)} \\
	    {((A\cp B)\times(A\cp B))\times (A\cp B)} && {A\times B} \\
	    {(A\cp B)\times(A\cp B)} & {A\cp B} & {(A\cp B)\times(A\cp B)}
	    \arrow["{\Delta\times\Delta}"', from=1-1, to=2-1]
	    \arrow["\Delta"', from=1-2, to=1-1]
	    \arrow["{(t_1\times u_1)\times (t_2\times u_2)}"', from=2-1, to=3-1]
	    \arrow["{(t_1\times t_2)\times (u_1\times u_2)}", from=2-1, to=3-3,sloped]
	    \arrow["{(i_1\times i_2)\times (i_1\times i_2)}"', from=3-1, to=4-1]
	    \arrow["{\inclusion\times\inclusion}", dashed, from=3-1, to=5-1,out=-10,in=10,pos=0.3,looseness=1.5]
	    \arrow["\simeq", dashed, from=3-3, to=3-1]
	    \arrow["{\sumhat A\times\sumhat B}", from=3-3, to=4-3]
	    \arrow["{\sumcp\times\sumcp}"', from=4-1, to=5-1]
	    \arrow["\inclusion", dashed, from=4-3, to=5-2,sloped]
	    \arrow["{i_1\times i_2}", from=4-3, to=5-3]
	    \arrow["\sumcp"', from=5-1, to=5-2]
	    \arrow["\sumcp", from=5-3, to=5-2]
	  \end{tikzcd}
	\]

	\begin{itemize}
	  \item The upper diagram commutes by coherence.
	  \item The bottom-left diagram commutes by functoriality of $\times$ and \cref{lem:inclusion}.
	  \item The bottom-right diagram commutes by \cref{lem:inclusion}.
	  \item The bottom-middle diagram commutes since
	    \[
	      (a,b)\sumcp(a',b') = (a\sumhat A a',b\sumhat B b')
	      \qedhere
	    \]
	\end{itemize}
    \end{itemize}
  \end{proof}

  \subsection[Theorem~\ref{thm:adequacy}]{\cref{thm:adequacy}}\label{B.3}
  \restate{Theorem}{thm:adequacy}{Adequacy}{
    If $\sem{\vdash t:A} = \sem{\vdash u:A}$, then $t\sim u$.
  }
  \begin{proof}
    By induction on the structure of $A$.
    \begin{itemize}
      \item Let $A=\top$. Then $t\sim u$, since there is no elimination context such that $[\cdot]:\top\vdash K:\top\vee\top$.
      \item Let $A=\bot$. This is impossible, since there is no closed proof of $\bot$.
      \item Let $A=B\Rightarrow C$, then, by \cref{thm:SR,thm:SN,thm:IP}, $t\longrightarrow^*\lambda x.t'$ and $u\longrightarrow\lambda x.u'$.
	Hence, since by \cref{thm:soundness}, $\sem{\vdash\lambda x.t':B\Rightarrow C}=\sem{\vdash t:A}=\sem{\vdash u:A}=\sem{\vdash\lambda x.u':B\Rightarrow C}$, we have 
	\[
	  \begin{tikzcd}[column sep=1.5cm]
	    \{\star\}\ar[r,"{\eta^B}"] & \home B{B\times\{\star\}}\ar[r,"\home B{t'}",bend left]\ar[r,"\home B{u'}"',bend right] & \home BC
	  \end{tikzcd}
	\]
	Thus, $\sem{x:B\vdash t':C}=\sem{x:B\vdash u':C}$.

	By \cref{lem:substitution}, for all $\vdash v:B$, $\sem{\vdash (v/x)t':C}=\sem{\vdash (v/x)u':C}$. Thus, by the induction hypothesis, 
	\begin{equation}
	  \label{eq:IHlambdaadequacyAlg}
	  (v/x)t'\sim(v/x)u'
	\end{equation}

	In addition, we have that for any $[\cdot]:B\Rightarrow C\vdash K:\top\vee\top$, $K = K'[[\cdot]~v]$ for some $\vdash v:B$.
	So, $K[\lambda x.t'] = K'[(\lambda x.t')~v]\longrightarrow K'[(v/x)t']$ and $K[\lambda x.u'] = K'[ (\lambda x.u')~v]\longrightarrow K'[(v/x)u']$.
	By~\eqref{eq:IHlambdaadequacy}, there exists $\vdash w:\top\vee\top$ such that $K'[(v/x)t']\lra^* w$ and $K'[(v/x)u']\lra^* w$.
	Therefore, 
	\[
	  K[\lambda x.t'] \longrightarrow K'[(v/x)t']\longrightarrow^* w\lla^* K'[(v/x)u']\lla K[\lambda x.u'] 
	\]
	Hence, $t\sim u$.

      \item Let $A=B\wedge C$. Then, by \cref{thm:SR,thm:SN,thm:IP}, $t\longrightarrow^*\pair{t_1}{t_2}$ and $u\longrightarrow^*\pair{u_1}{u_2}$.
	Hence, since by \cref{thm:soundness}, $\sem{\vdash\pair{t_1}{t_2}:B\wedge C}=\sem{\vdash t:A}=\sem{\vdash u:A}=\sem{\vdash\pair{u_1}{u_2}:B\wedge C}$, we have
	\[
	  \begin{tikzcd}[column sep=1.5cm]
	    \{\star\}\ar[r,"\Delta"] & \{\star\}\times\{\star\}\ar[r,"t_1\times t_2",bend left]\ar[r,"u_1\times u_2"',bend right] & B\times C
	  \end{tikzcd}
	\]
	Thus, $\sem{\vdash t_1:B}=\sem{\vdash u_1:B}$ and $\sem{\vdash t_2:C}=\sem{\vdash u_2:C}$, and so, by the induction hypothesis $t_1\sim u_1$ and $t_2\sim u_2$. Hence, $t\sim u$.

      \item Let $A=B\vee C$. Then, by \cref{thm:SR,thm:SN,thm:IP},
	\[
	  t\lra^*\inl(t')\text{,}\quad
	  t\lra^*\inr(t')\text{,}\quad\text{or}\quad 
	  t\lra^*\inl(t')\plus\inr(t'')\text{,}
	\]
	and
	\[
	  u\lra^*\inl(u')\text{,}\quad
	  u\lra^*\inr(u')\text{,}\quad\text{or}\quad
	  u\lra^*\inl(u')\plus\inr(u'')\text{.}
	\]

	However, by \cref{thm:soundness}, and the fact that
	$\sem{\vdash t:A}=\sem{\vdash u:A}$, we have that
	\begin{itemize}
	  \item if $t\lra^*\inl(t')$ then $u\lra^*\inl(u')$, 
	  \item if $t\lra^*\inr(t')$ then $u\lra^*\inr(u')$, and 
	  \item if $t\lra^*\inl(t')\plus\inr(t'')$ then $u\lra^*\inl(u')\plus\inr(u'')$.
	\end{itemize}
	We must consider these three cases. 
	\begin{itemize}
	  \item In the first case, we have
	    \[
	      \begin{tikzcd}[column sep=1.5cm]
		\{\star\}\ar[r,"t'",bend left,pos=0.45]\ar[r,"u'"',bend right,pos=0.45] &A \ar[r,"i_1"] & A\cp B \\
	      \end{tikzcd}
	    \]
	    Thus, $\sem{\vdash t':A} = \sem{\vdash u':A}$, and so, by the induction hypothesis $t'\sim u'$. Hence, $t\sim u$.
	  \item In the second case, we have
	    \[
	      \begin{tikzcd}[column sep=1.5cm]
		\{\star\}\ar[r,"t'",bend left,pos=0.45]\ar[r,"u'"',bend right,pos=0.45] &B \ar[r,"i_2"] & A\cp B \\
	      \end{tikzcd}
	    \]
	    Thus, $\sem{\vdash t':B} = \sem{\vdash u':B}$, and so, by the induction hypothesis $t'\sim u'$. Hence, $t\sim u$.
	  \item In the third case, we have
	    \[
	      \begin{tikzcd}[column sep=1.5cm,row sep=1cm]
		\{\star\}\ar[r,"\Delta"] & \{\star\}\times\{\star\}\ar[r,"t'\times t''",bend left]\ar[r,"u'\times u''"',bend right] &[-1cm]A\times B \ar[d,"i_1\times i_2"] \\
		A\cp B&& (A\cp B)\times (A\cp B)\ar[ll,"\sumhat\times"]
	      \end{tikzcd}
	    \]
	    Thus, $\sem{\vdash t':A} = \sem{\vdash u':A}$ and $\sem{\vdash t'':B} = \sem{\vdash u'':B}$, and so, by the induction hypothesis $t'\sim u'$ and $t''\sim u''$. Hence, $t\sim u$.
	    \qedhere
	\end{itemize}
    \end{itemize}
  \end{proof}

  \section{Admissibility of commutation between the parallel operator and the elimination of disjunction}\label{C}
  The rule $\elimor(t\plus u,\abstr{x}v,\abstr{y}w)  \longrightarrow \elimor(t,\abstr{x}v,\abstr{y}w) \plus \elimor(u,\abstr{x}v,\abstr{y}w)$ could have been chosen instead of the rules implementing the commutation of the parallel operator with the elimination of disjunction. We show in this section that this rule is admissible in our model.

  The commuting diagram is the following
  \[
    \begin{tikzcd}[column sep=1cm,row sep=5mm]
      \Gamma
      \arrow[rdddddd, "{\elimor(t\plus u,\abstr{x}v,\abstr{y}w)}",sloped,
	rounded corners,
	to path={[pos=0.25]
	  -- ([yshift=5mm]\tikztostart.north)
	  -| ([xshift=1.7cm]\tikztotarget.east)\tikztonodes
	-- (\tikztotarget)}
      ]
      \arrow[rdddddd, "{\elimor(t,\abstr{x}v,\abstr{y}w) \plus \elimor(u,\abstr{x}v,\abstr{y}w)}"',sloped,
	rounded corners,
	to path={[pos=0.75]
	  -- ([xshift=-3.1cm]\tikztostart.west)
	  |- ([yshift=-5mm]\tikztotarget.south)\tikztonodes
	-- (\tikztotarget)}
      ]
      & {\Gamma\times\Gamma} \\
      {\Gamma\times\Gamma\times \Gamma\times\Gamma} & {\Gamma\times\Gamma\times\Gamma} \\
      {A\times\Gamma\times B\times\Gamma} & {A\times B\times\Gamma} \\
      & {(A\cp B)\times(A\cp B)\times\Gamma} \\
      {((A\cp B)\times\Gamma)\times((A\cp B)\times\Gamma)} & {(A\cp B)\times\Gamma} \\
      {((A\times\Gamma)\cp(B\times\Gamma))\times((A\times\Gamma)\cp(B\times\Gamma))} & {(A\times\Gamma)\cp(B\times\Gamma)} \\
      {C\times C} & C 
      \arrow["\Delta", from=1-1, to=1-2]
      \arrow["{\Delta\times\Delta}", from=1-2, to=2-1,sloped]
      \arrow["{\Delta\times\Id}", from=1-2, to=2-2]
      \arrow["{t\times\Id\times u\times\Id}"', from=2-1, to=3-1]
      \arrow["{t\times u\times\Id}", from=2-2, to=3-2]
      \arrow["{i_1\times\Id\times i_2\times\Id}"', from=3-1, to=5-1]
      \arrow["\delta"', dashed, from=3-2, to=3-1]
      \arrow["{i_1\times i_2\times\Id}", from=3-2, to=4-2]
      \arrow["{\sumcp\times\Id}", from=4-2, to=5-2]
      \arrow["{d\times d}"', from=5-1, to=6-1]
      \arrow["d", from=5-2, to=6-2]
      \arrow["{\coproducto vw\times\coproducto vw}"', from=6-1, to=7-1]
      \arrow["{\coproducto vw}", from=6-2, to=7-2]
      \arrow["{\sumhat C}"', from=7-1, to=7-2]
    \end{tikzcd}
  \]
  where the top diagram commutes by coherence and the bottom diagram commutes
  since 
  \[
    \adjustbox{valign=b}{
      \begin{tikzcd}[column sep=2cm,row sep=8mm]
	((a, g), (b, g)) & ((a,b),g) \\
	& ((a,0),(b,1), g) \\
	(((a,0), g),((b,1), g)) & ((a,b), g) \\
	(((a,g),0),((b, g),1)) & ((a, g),(b, g)) \\
	(v(a,g),w(b,g)) &  v(a,g)\sumhat C w(b,g)
	\arrow[maps to,"\delta"', dashed, from=1-2, to=1-1]
	\arrow[maps to,"{i_1\times\Id\times i_2\times\Id}"', from=1-1, to=3-1]
	\arrow[maps to,"{i_1\times i_2\times\Id}", from=1-2, to=2-2]
	\arrow[maps to,"{\sumcp\times\Id}", from=2-2, to=3-2]
	\arrow[maps to,"{d\times d}"', from=3-1, to=4-1]
	\arrow[maps to,"d", from=3-2, to=4-2]
	\arrow[maps to,"{\coproducto vw\times\coproducto vw}"', from=4-1, to=5-1]
	\arrow[maps to,"{\coproducto vw}", from=4-2, to=5-2]
	\arrow[maps to,"{\sumhat C}"', from=5-1, to=5-2]
      \end{tikzcd}
    }
  \]

  \section[Full proofs from Section~\ref{sec:modelAlg}]{Full proofs from \cref{sec:modelAlg}}\label{D}

  \subsection[Lemma~\ref{lem:substitutionAlg}]{\cref{lem:substitutionAlg}}\label{D.1}
  \restate{Lemma}{lem:substitutionAlg}{Substitution}
  {
    If $x:A,\Gamma\vdash t:B$ and $\Gamma\vdash u:A$, then
    the following diagram commutes.
  }
  \[
    \begin{tikzcd}[column sep=1cm]
      \sem\Gamma\ar[r,"(u/x)t"]\ar[d,"\Delta"'] & \sem B \\
      \sem\Gamma\times\sem\Gamma\ar[r,"u\times\Id"'] & \sem A\times\sem\Gamma\ar[u,"t"']
    \end{tikzcd}
  \]
  \begin{proof}
    By induction on $t$. 
    We only prove the cases of $\sstar s$, $\elimtop(t,u)$, and $\scal s\bullet t$, since all the other cases are analogous to those in the proof of \cref{lem:substitution}.
    \begin{itemize}
      \item Case $t = \sstar s$. Then $(u/x)t=\sstar s$ and $B=\top$. The diagram is
	\[
	  \begin{tikzcd}
	    \Gamma && {\mathcal S} \\
	    & {\{\star\}} \\
	    {\Gamma\times\Gamma} && {A\times\Gamma}
	    \arrow["{(u/x)\sstar s}", from=1-1, to=1-3]
	    \arrow["{!}", from=1-1, to=2-2,dashed,sloped]
	    \arrow["\Delta"', from=1-1, to=3-1]
	    \arrow["\scal s", from=2-2, to=1-3,dashed,sloped]
	    \arrow["{u\times \Id}"', from=3-1, to=3-3]
	    \arrow["{\sstar s}"', from=3-3, to=1-3]
	    \arrow["{!}", from=3-3, to=2-2,dashed,sloped]
	  \end{tikzcd}
	\]
	\begin{itemize}
	  \item The upper and the right diagrams are the definitions.
	  \item The left diagram commutes since $\{\star\}$ is a terminal object.
	\end{itemize}
      \item Case $t=\elimtop(t_1,t_2)$. The diagram is
	\[
	  \begin{tikzcd}[column sep=1cm]
	    \Gamma &[-1.1cm]&[1.8cm]&[-1.1cm] B \\[-2mm]
	    & {\Gamma\times\Gamma} & {\mathcal S\times B} \\
	    & {(\Gamma\times\Gamma)\times(\Gamma\times\Gamma)} & {(A\times\Gamma)\times(A\times\Gamma)} \\[-2mm]
	    {\Gamma\times\Gamma} &&& {A\times\Gamma}
	    \arrow["{(u/x)\elimtop(t_1,t_2)}", from=1-1, to=1-4]
	    \arrow["\Delta", from=1-1, to=2-2,dashed,sloped]
	    \arrow["\Delta"', from=1-1, to=4-1]
	    \arrow["{(u/x)t_1\times (u/x)t_2}", from=2-2, to=2-3,dashed]
	    \arrow["{\Delta\times\Delta}"', from=2-2, to=3-2,dashed]
	    \arrow["{\prodhat B}", from=2-3, to=1-4,dashed,sloped]
	    \arrow["{(u\times\Id)\times(u\times\Id)}", from=3-2, to=3-3,dashed]
	    \arrow["{t_1\times t_2}"', from=3-3, to=2-3,dashed]
	    \arrow["\Delta", from=4-1, to=3-2,dashed,sloped]
	    \arrow["{u\times\Id}"', from=4-1, to=4-4]
	    \arrow["{\elimtop(t_1,t_2)}"', from=4-4, to=1-4]
	    \arrow["\Delta", from=4-4, to=3-3,dashed,sloped]
	  \end{tikzcd}
	\]
	\begin{itemize}
	  \item The upper and the right diagrams are the definitions.
	  \item The left and the bottom diagram commutes by the naturality of $\Delta$.
	  \item The central diagram commutes by the induction hypothesis and the functoriality of $\times$.
	\end{itemize}
      \item Case $t=\scal s\bullet t'$. The diagram is
	\[
	  \begin{tikzcd}[column sep=1.5cm]
	    \Gamma && B \\
	    & B \\
	    {\Gamma\times\Gamma} && {A\times\Gamma}
	    \arrow["{(u/x)(\scal s\bullet t')}", from=1-1, to=1-3]
	    \arrow["{(u/x)t}", dashed,sloped, from=1-1, to=2-2]
	    \arrow["\Delta"', from=1-1, to=3-1]
	    \arrow["{\hat{\scal s}}", dashed,sloped, from=2-2, to=1-3]
	    \arrow["{u\times \Id}", from=3-1, to=3-3]
	    \arrow["{\scal s\bullet t}"', from=3-3, to=1-3]
	    \arrow["t", dashed,sloped, from=3-3, to=2-2]
	  \end{tikzcd}
	\]
	\begin{itemize}
	  \item The upper and the right diagrams are the definitions.
	  \item The left diagram commutes by the induction hypothesis.
	    \qedhere
	\end{itemize}
    \end{itemize}
  \end{proof}

  \subsection[{Theorem~\ref{thm:soundnessAlg}}]{\cref{thm:soundnessAlg}}\label{D.2}
  \restate{Theorem}{thm:soundnessAlg}{Soundness}
  {
    If $t\lra r$ and $\Gamma\vdash t:A$, then
    $\sem{\Gamma\vdash t:A} = \sem{\Gamma\vdash r:A}$.
  }
  \begin{proof}
    By induction on the relation $\longrightarrow$.
    We only prove the cases that differ from the reduction rules of $\lambda_\parallel$, since all the other cases are analogous to those in the proof of \cref{thm:soundness}.
    \begin{itemize}
      \item Rule $\elimtop(\sstar s,t) \longrightarrow  \scal s\bullet t$. The commuting diagram is
	\[
	  \begin{tikzcd}
	    \Gamma & {\Gamma\times \Gamma} & {\{\star\}\times A} \\
	    A & A & {\mathcal S\times A}
	    \arrow["\Delta", from=1-1, to=1-2]
	    \arrow["t"', from=1-1, to=2-1]
	    \arrow["{{!}\times t}", from=1-2, to=1-3]
	    \arrow["{\scal s\times \Id}", from=1-3, to=2-3]
	    \arrow["\rho", dashed, from=2-1, to=1-3,sloped]
	    \arrow["{\hat{\scal s}}"', from=2-1, to=2-2]
	    \arrow["{\prodhat A}", from=2-3, to=2-2]
	  \end{tikzcd}
	\]
	\begin{itemize}
	  \item The commutation of the left diagram is trivial by the definition of $\rho$.
	  \item The right diagram is the definition of $\hat{\scal s}$.
	\end{itemize}
      \item Rule $\sstare{\scal s_1}\plus \sstare{\scal s_2} \longrightarrow  \sstare{\scal s_1+\scal s_2}$. The commuting diagram is
	\[
	  \begin{tikzcd}
	    \Gamma & {\{\star\}} & {\mathcal S} \\
	    {\Gamma\times\Gamma} & {\{\star\}\times \{\star\}} & {\mathcal S\times \mathcal S}
	    \arrow["{{!}}", from=1-1, to=1-2]
	    \arrow["\Delta"', from=1-1, to=2-1]
	    \arrow["{\scal s_1+\scal s_2}", from=1-2, to=1-3]
	    \arrow["\Delta"', from=1-2, to=2-2,dashed]
	    \arrow["{{!}\times{!}}"', from=2-1, to=2-2]
	    \arrow["{\scal s_1\times\scal s_2}"', from=2-2, to=2-3]
	    \arrow["\sumhat{\mathcal S}"', from=2-3, to=1-3]
	  \end{tikzcd}
	\]
	\begin{itemize}
	  \item The left diagram commutes by naturality of $\Delta$.
	  \item The right diagram commutes since the magma operation of $\mathcal S$ is defined with respect to the bi-magma operation: $\sumhat{\mathcal S}=+$.
	\end{itemize}
      \item Rule $\scal s_1\bullet\sstare{\scal s_2} \longrightarrow  \sstare{\scal s_1\cdot s_2}$. The commuting diagram is
	\[
	  \begin{tikzcd}[column sep=1.5cm]
	    \Gamma & {\{\star\}} && {\mathcal S} \\
	    && {\mathcal S}
	    \arrow["{{!}}", from=1-1, to=1-2]
	    \arrow["{\scal s_1\cdot\scal s_2}", from=1-2, to=1-4]
	    \arrow["{\scal s_2}"', from=1-2, to=2-3,sloped]
	    \arrow["{\hat{\scal s}_1}"', from=2-3, to=1-4,sloped]
	  \end{tikzcd}
	\]
	The diagram commutes since the action of $\mathcal S$ is defined with respect to the bi-magma operation: $\prodhat{\mathcal S}=\cdot$.
      \item Rule $\scal s\bullet \lambda \abstr{x}t \longrightarrow  \lambda \abstr{x}\scal s\bullet t$. The commuting diagram is
	\[
	  \begin{tikzcd}[column sep=1cm]
	    \Gamma & {\home{A}{A\times\Gamma}} && {\home AB} & {\home AB}
	    \arrow["{\eta^A}", from=1-1, to=1-2]
	    \arrow["{\home At}", from=1-2, to=1-4]
	    \arrow["{\hat{\scal s}}"', bend right, from=1-4, to=1-5]
	    \arrow["{\home A{\hat{\scal s}}}", bend left, from=1-4, to=1-5]
	  \end{tikzcd}
	\]
	The diagram commutes since $(\scal s\prodhat{\home AB}f)(a) =\scal s\prodhat B f(a)$ by definition.
      \item Rule $\scal s\bullet\pair{t}{v} \longrightarrow  \pair{\scal s\bullet t}{\scal s\bullet v}$. The commuting diagram is
	\[
	  \begin{tikzcd}[column sep=1.5cm]
	    \Gamma & {\Gamma\times\Gamma} & {A\times B} & {A\times B}
	    \arrow["\Delta", from=1-1, to=1-2]
	    \arrow["{t\times v}", from=1-2, to=1-3]
	    \arrow["{\hat{\scal s}}"',bend right, from=1-3, to=1-4]
	    \arrow["{\hat{\scal s}\times\hat{\scal s}}", bend left, from=1-3, to=1-4]
	  \end{tikzcd}
	\]
	The diagram commutes since $\scal s\prodhat{A\times B}(a,b) = (\scal s\prodhat A a,\scal s\prodhat B b)$ by definition.
      \item Rule $\scal s\bullet\inl(t) \longrightarrow \inl(\scal s\bullet t)$. The commuting diagram is
	\[
	  \begin{tikzcd}[column sep=1.5cm]
	    \Gamma & A & {A\cp B} \\
	    & A & {A\cp B}
	    \arrow["t", from=1-1, to=1-2]
	    \arrow["{i_1}", from=1-2, to=1-3]
	    \arrow["{\hat{\scal s}}"', from=1-2, to=2-2]
	    \arrow["{\hat{\scal s}}", from=1-3, to=2-3]
	    \arrow["{i_1}"', from=2-2, to=2-3]
	  \end{tikzcd}
	\]
	The diagram commutes since $\scal s\prodcp(a,0) = (\scal s\prodhat A a,0)$.
      \item Rule $\scal s\bullet\inr(t)\longrightarrow \inr(\scal s\bullet t)$. The commuting diagram is
	\[
	  \begin{tikzcd}[column sep=1.5cm]
	    \Gamma & B & {A\cp B} \\
	    & B & {A\cp B}
	    \arrow["t", from=1-1, to=1-2]
	    \arrow["{i_2}", from=1-2, to=1-3]
	    \arrow["{\hat{\scal s}}"', from=1-2, to=2-2]
	    \arrow["{\hat{\scal s}}", from=1-3, to=2-3]
	    \arrow["{i_2}"', from=2-2, to=2-3]
	  \end{tikzcd}
	\]
	The diagram commutes since $\scal s\prodcp(b,1) = (\scal s\prodhat B b,0)$.

      \item Rule $\scal s\bullet(t\plus v) \longrightarrow \scal s\bullet t\plus \scal s\bullet v$. The commuting diagram is
	\[
	  \begin{tikzcd}
	    {\Gamma\times\Gamma} & \Gamma \\
	    {A\times B} & {(A\cp B)\times (A\cp B)} \\
	    {A\times B} & {A\cp B} \\
	    {(A\cp B)\times (A\cp B)} & {A\cp B}
	    \arrow["{t\times v}", from=1-1, to=2-1]
	    \arrow["\Delta", from=1-2, to=1-1]
	    \arrow["{i_1\times i_2}", from=2-1, to=2-2]
	    \arrow["{\hat{\scal s}\times\hat{\scal s}}"', from=2-1, to=3-1]
	    \arrow["inc", from=2-1, to=3-2,sloped]
	    \arrow["\sumcp", from=2-2, to=3-2]
	    \arrow["{i_1\times i_2}"', from=3-1, to=4-1]
	    \arrow["inc", from=3-1, to=4-2,sloped]
	    \arrow["{\hat{\scal s}}", from=3-2, to=4-2]
	    \arrow["\sumcp"', from=4-1, to=4-2]
	  \end{tikzcd}
	\]
	The diagram commutes since $\hat{\scal s}\times\hat{\scal s} = \hat{\scal s}$.
	\qedhere
    \end{itemize}
  \end{proof}

  \subsection[Theorem~\ref{thm:adequacyAlg}]{\cref{thm:adequacyAlg}}\label{D.3}
  \restate{Theorem}{thm:adequacyAlg}{Adequacy}{
    If $\sem{\vdash t:A} = \sem{\vdash u:A}$, then $t\sim u$.
  }
  \begin{proof}
    By induction on the structure of $A$.
    \begin{itemize}
      \item Let $A=\top$, then, $t\longrightarrow^*\sstare{\scal s_1}$ and $u\longrightarrow^*\sstare{\scal s_2}$.
	Hence, since by \cref{thm:soundnessAlg}, $\sem{\vdash \sstare{\scal s_1}:\top}=\sem{\vdash t:A}=\sem{\vdash u:A}=\sem{\vdash \sstare{\scal s_2}:\top}$, we have
	\[
	  \begin{tikzcd}[column sep=1.5cm]
	    \{\star\}\ar[r,"!"] & \{\star\}\ar[r,"\scal s_1",pos=0.45,bend left]\ar[r,"\scal s_2"',pos=0.45,bend right] & \mathcal S
	  \end{tikzcd}
	\]
	Thus, $\scal s_1=\scal s_2$. Hence, $t\sim u$.

      \item Let $A=\bot$. This is impossible, since there is no closed proof of $\bot$.

      \item Let $A=B\Rightarrow C$, then, $t\longrightarrow^*\lambda x.t'$ and $u\longrightarrow\lambda x.u'$.
	Hence, since by \cref{thm:soundnessAlg}, $\sem{\vdash\lambda x.t':B\Rightarrow C}=\sem{\vdash t:A}=\sem{\vdash u:A}=\sem{\vdash\lambda x.u':B\Rightarrow C}$, we have 
	\[
	  \begin{tikzcd}[column sep=1.5cm]
	    \{\star\}\ar[r,"{\eta^B}"] & \home B{B\times\{\star\}}\ar[r,"\home B{t'}",bend left]\ar[r,"\home B{u'}"',bend right] & \home BC
	  \end{tikzcd}
	\]
	Thus, $\sem{x:B\vdash t':C}=\sem{x:B\vdash u':C}$.

	By \cref{lem:substitutionAlg}, for all $\vdash v:B$, $\sem{\vdash (v/x)t':C}=\sem{\vdash (v/x)u':C}$. Thus, by the induction hypothesis, 
	\begin{equation}
	  \label{eq:IHlambdaadequacy}
	  (v/x)t'\sim(v/x)u'
	\end{equation}

	In addition, we have that for any $[\cdot]:B\Rightarrow C\vdash K:\top$, $K = K'[[\cdot]~v]$ for some $\vdash v:B$.
	So, $K[\lambda x.t'] = K'[(\lambda x.t')~v]\longrightarrow K'[(v/x)t']$ and $K[\lambda x.u'] = K'[ (\lambda x.u')~v]\longrightarrow K'[(v/x)u']$.
	By~\eqref{eq:IHlambdaadequacy}, there exists $\vdash v:\top$ such that $K'[(v/x)t']\lra^* v$ and $K'[(v/x)u']\lra^* v$.
	Therefore, 
	\begin{align*}
	  K[\lambda x.t'] & \longrightarrow K'[(v/x)t']\longrightarrow^* v\\
	  K[\lambda x.u'] & \longrightarrow K'[(v/x)u']\longrightarrow^* v
	\end{align*}
	Hence, $t\sim u$.

      \item Let $A=B\wedge C$, then, $t\longrightarrow^*\pair{t_1}{t_2}$ and $u\longrightarrow^*\pair{u_1}{u_2}$.
	Hence, since by \cref{thm:soundnessAlg}, $\sem{\vdash\pair{t_1}{t_2}:B\wedge C}=\sem{\vdash t:A}=\sem{\vdash u:A}=\sem{\vdash\pair{u_1}{u_2}:B\wedge C}$, we have
	\[
	  \begin{tikzcd}[column sep=1.5cm]
	    \{\star\}\ar[r,"\Delta"] & \{\star\}\times\{\star\}\ar[r,"t_1\times t_2",bend left]\ar[r,"u_1\times u_2"',bend right] & B\times C
	  \end{tikzcd}
	\]
	Thus, $\sem{\vdash t_1:B}=\sem{\vdash u_1:B}$ and $\sem{\vdash t_2:C}=\sem{\vdash u_2:C}$, and so, by the induction hypothesis $t_1\sim u_1$ and $t_2\sim u_2$. Hence, $t\sim u$.

      \item Let $A=B\vee C$, then,
	$t\lra^*\inl(t')$,
	$t\lra^*\inr(t')$, or 
	$t\lra^*\inl(t')\plus\inr(t'')$,
	and
	$u\lra^*\inl(u')$,
	$u\lra^*\inr(u')$, or
	$u\lra^*\inl(u')\plus\inr(u'')$.

	By \cref{thm:soundness}, and since
	$\sem{\vdash t:A}=\sem{\vdash u:A}$, we have that if
	$t\lra^*\inl(t')$ then $u\lra^*\inl(u')$, if
	$t\lra^*\inr(t')$ then $u\lra^*\inr(u')$, and if
	$t\lra^*\inl(t')\plus\inr(t'')$ then
	$u\lra^*\inl(u')\plus\inr(u'')$.

	We consider the three cases.
	\begin{itemize}
	  \item In the first case, we have
	    \[
	      \begin{tikzcd}[column sep=1.5cm]
		\{\star\}\ar[r,"t'",bend left,pos=0.45]\ar[r,"u'"',pos=0.45,bend right] &A \ar[r,"i_1"] & A\cp B \\
	      \end{tikzcd}
	    \]
	    Thus, $\sem{\vdash t':A} = \sem{\vdash u':A}$, and so, by the induction hypothesis $t'\sim u'$. Hence, $t\sim u$.
	  \item In the second case, we have
	    \[
	      \begin{tikzcd}[column sep=1.5cm]
		\{\star\}\ar[r,"t'",bend left,pos=0.45]\ar[r,"u'"',pos=0.45,bend right] &B \ar[r,"i_2"] & A\cp B \\
	      \end{tikzcd}
	    \]
	    Thus, $\sem{\vdash t':B} = \sem{\vdash u':B}$, and so, by the induction hypothesis $t'\sim u'$. Hence, $t\sim u$.
	  \item In the third case, we have
	    \[
	      \begin{tikzcd}[column sep=1.5cm,row sep=1cm]
		\{\star\}\ar[r,"\Delta"] & \{\star\}\times\{\star\}\ar[r,"t'\times t''",bend left]\ar[r,"u'\times u''"',bend right] &[-1cm]A\times B \ar[d,"i_1\times i_2"] \\
		A\cp B&& (A\cp B)\times (A\cp B)\ar[ll,"\sumhat\times"]
	      \end{tikzcd}
	    \]
	    Thus, $\sem{\vdash t':A} = \sem{\vdash u':A}$ and $\sem{\vdash t'':B} = \sem{\vdash u'':B}$, and so, by the induction hypothesis $t'\sim u'$ and $t''\sim u''$. Hence, $t\sim u$.
	    \qedhere
	\end{itemize}
    \end{itemize}
  \end{proof}

\fi
\end{document}